\documentclass[journal,onecolumn]{IEEEtran}
\usepackage[utf8]{inputenc}
\usepackage[T1]{fontenc}
\usepackage{url}
\usepackage{ifthen}
\usepackage{cite}
\usepackage[cmex10]{amsmath}
\usepackage{graphicx,epstopdf,amssymb,amsthm}
\usepackage{ulem,color}
\usepackage{epstopdf}
\newcommand{\ie}{\textit{i.e. }}

\usepackage[utf8]{inputenc}
\usepackage[english]{babel}
\usepackage{caption}
\DeclareCaptionLabelFormat{lc}{\MakeLowercase{#1}~#2}
\newtheorem{theorem}{Theorem}
\newtheorem{corollary}{Corollary}
\newtheorem*{remark}{Remark}

\newtheorem{definition}{Definition}

\begin{document}

\title{Coded Caching with Heterogeneous User Profiles}

\author{
\IEEEauthorblockN{Ciyuan Zhang,~\IEEEmembership{Student Member,~IEEE} Su Wang, Vaneet Aggarwal,~\IEEEmembership{Senior Member,~IEEE} and Borja Peleato,~\IEEEmembership{Senior Member,~IEEE}\thanks{C. Zhang, S. Wang, and V. Aggarwal are with Purdue University, West Lafayette, IN 47906, USA. Email:\{zhan3375,wang2506,vaneet\}@purdue.edu}}
\thanks{B. Peleato is with Carlos III University of Madrid, Leganes, Spain. Email: bpeleato@ing.uc3m.es}
\thanks{This work was partially funded by the CONEX-Plus Programme - Marie-Sklodowska Curie COFUND Action (H2020-MSCA-COFUND-2017- GA 801538) and Banco Santander.}}%
\bibliographystyle{IEEEtran}
\maketitle

\begin{abstract}
Coded caching utilizes pre-fetching during off-peak hours and multi-casting for delivery in order to balance the traffic load in communication networks. Several works have studied the achievable peak and average rates under different conditions: variable file lengths or popularities, variable cache sizes, decentralized networks, etc. However, very few have considered the possibility of heterogeneous user profiles, despite modern content providers are investing heavily in categorizing users according to their habits and preferences.

This paper proposes three coded caching schemes with uncoded pre-fetching for scenarios where end users are grouped into classes with different file demand sets (FDS). One scheme ignores the difference between the classes, another ignores the intersection between them and the third decouples the delivery of files common to all FDS from those unique to a single class. The transmission rates of the three schemes are compared with a lower bound to evaluate their gap to optimality, and with each other to show that each scheme can outperform the other two when certain conditions are met.


\end{abstract}

\section{Introduction}
\label{s-intro}

The recent information explosion is constantly pushing the limits of communication networks, users always want more information at faster speeds and with minimal latency. Network operators hope to address this problem by pushing the content and computation closer to the end users, in what is commonly known as fog networking~\cite{yi2015survey}. Having multiple caches distributed across the network helps balance the load over the internet backbone, but does not alleviate the congestion that often arises at the edge of the network during peak hours. Coded caching was introduced as a powerful solution for solving this problem.

A coded caching scheme consists of a placement and a delivery phase. The placement phase takes place during off-peak hours, when there are spare resources in the network. The server partitions all the files into segments and stores them in the users' caches. The delivery phase takes place during peak hours, when multiple (if not all) users have file requests. The server attempts to fulfill all those requests with minimal information transmitted, by leveraging the segments cached during the placement phase. It has long been known that proactively caching popular content during off-peak hours reduces the total information to be transmitted when that content is requested. This gain depends on the hit rates on the local cache of the end users, so it is known as local caching gain. However, Maddah-Ali and Niesen's seminal paper~\cite{maddah2014fundamental} recently showed that the overall transmission rate in
point-to-multipoint links can be reduced further by carefully coordinating the cached segments and using a coded delivery scheme. This gain depends on the segments shared by the different user subgroups and is therefore known as global caching gain.

In~\cite{maddah2014fundamental}, Maddah-Ali and Niesen proposed a coded caching scheme which maximizes multicasting opportunities for the worst case user demands. Subsequent works focused on lowering the peak rate in different scenarios~\cite{niesen2017coded, pedarsani2016online}. However, these papers adopted homogeneous models which do not fit most practical systems where coded caching could potentially be used.
%
%
Some recent works have analyzed the transmission rate of coded caching systems with full heterogeneity:~\cite{chan2019coded} considered different file sizes, cache sizes, and user dependent file popularity, but only for two users and two files. Centralized and decentralized coded caching schemes with heterogeneous user cache sizes were studied in~\cite{ibrahim2019coded}~and~\cite{bayat2020cache}, respectively, but they ignored the users' diverse preference over files. The work in~\cite{daniel2019optimization} provided an optimization theoretic analysis of coded caching systems with various heterogeneities (cache size, file length, and file popularity) and demonstrated that Maddah-Ali and Niesen's original scheme from~\cite{maddah2014fundamental} is optimal for problems with uniform file size, popularity, and cache size.
Unfortunately, the results in~\cite{daniel2019optimization} are derived numerically, without theoretical evidence. Furthermore, it does not address user heterogeneity, which is the focus of this paper. A scenario with heterogeneous file popularities was addressed in \cite{8091300} by partitioning files into groups such that, within each group, the files have approximately equal popularities. However, it again assumed that the popularity of each file was identical for every user. Additional research on coded caching has extended it to topics such as device-to-device caching~\cite{golrezaei2013femtocaching}, hierarchical caching~\cite{karamchandani2016hierarchical}, and distinct file sizes~\cite{zhang2015coded}.

Papers like~\cite{hannak2013measuring,mccallum1999machine,agichtein2006learning} have addressed the significance of predicting users behavior according to their preferences. This mirrors the current trend of online video streaming companies like Hulu and Netflix which spend a considerable amount of resources investigating their customers' habits and categorizing them according to their streaming preferences.
The paper~\cite{lu2019effective} utilized game theory to analyze the transmission cost of a centralized coded caching system when the users present heterogeneous preferences over the files requested, but it neglected the alternative of organizing the users into groups according to their preferences. The recent paper~\cite{he2020coded} categorized users into two groups: VIP and non-VIP. It proposed schemes so that the VIP group obtains better experience (lower transmission rate) than the other group. However, it only considered the decentralized coded caching system model and merely paid attention to specific users instead of the whole ensemble.
Another coded caching scheme with user grouping was proposed in~\cite{tegin2020coded} for wireless channels, and its results indicate that grouping users based on their channel conditions is beneficial for reducing transmission time, especially for small cache sizes.
Our prior conference paper~\cite{wang2019coded} addressed a system where the users are grouped into classes with similar file interests. It proposed three coded caching schemes for this scenario and studied their peak rate, but it did not provide a comprehensive comparison between them. This paper will do that and study the subject in more detail.

Most existing works have focused on studying the peak rather than the average rate of coded caching systems. This is mainly due to the fact that the average rate is highly dependent on the distribution of the requests and that the peak rate is an important factor in the design of small networks. When the number of files and users is large, however, the peak rate is very rarely reached and the average rate is a better metric for performance evaluation. There have been works studying the average rate, e.g.,~\cite{luo2019coded,yu2018exact}, but they assumed the same distribution of requests for all users. The scenario with heterogeneous user profiles was thoroughly analyzed in~\cite{chang2020twousers}, but just for the case of two users. Our prior work~\cite{zhang202averagerate} studied the average rate resulting from the three schemes proposed in~\cite{wang2019coded} and compared their asymptotic performance.

The main contributions of this paper include:
1) characterizing the peak and average rates of the three schemes proposed in~\cite{wang2019coded} for a coded caching system with heterogeneous user profiles; 2) deriving lower bounds for the peak rate of the three schemes; 3) proposing a cache distribution method which results in minimal peak and average rate for one of the schemes when the caches are relatively small compared with the size of the library; 4) comparing the peak rate of transmission of the three schemes analytically to provide insights for deciding which scheme to choose given the system's parameters.
%
%
%

The paper will be organized as follows: Section~\ref{s-background} introduces our system model and the notation to be used throughout the paper. Section~\ref{s-schemes} describes the three coded caching schemes being proposed and analyzed. Section~\ref{s-lower bound} derives a lower bound for the peak rate of a coded caching scheme with heterogeneous user profiles and compares the peak rate of the three schemes with that bound. Section~\ref{s-Results} studies how to optimally distribute the cache among the different types of files and compares the peak rate of the three schemes according to the cache size. Finally, Section~\ref{s-simulations} provides numerical simulation results to illustrate and support our derivations, and Section~\ref{s-conclusion} concludes the paper.


\section{Background}
\label{s-background}

\subsection{System Model}

This paper considers a system with a single server storing $N$ files of size $F$, which is connected through an error-free broadcast link to $K$ end users equipped with cache memories of size $MF$ each. The $K$ users are split into $G$ classes according to the files that they may request. Out of the $N$ files, $N_c$ are common files which may be requested by users in any class and the rest are unique files which are only appealing to one class of users, such as cartoon or sci-fi movies. For simplicity, this paper assumes that the number of users and unique files is the same for every class, and that they do not intersect. Therefore, each class has $\frac{K}{G}$ users and $N_u$ unique files, where $N=N_c+GN_u$. Furthermore, we assume that the number of users is smaller than the number of common and unique files. The quantities $K$, $G$, $N_c$, and $N_u$ are generally discrete in practice, but this paper will often treat them as continuous to avoid integer effects during calculations. If their values are large enough, the rounding errors can be neglected. This scenario is illustrated in Fig.~\ref{fig:sys_model} with only two classes.

In the placement phase, caches are populated with file segments. This paper only considers uncoded prefetching, which means that segments are cached in plain form, not coded together. As asserted in~\cite{chang2020twousers}, uncoded prefetching is suboptimal, but it has many advantages: it allows for asynchronous transmissions, reduces latency, simplifies the bookeeping, etc. Since file segments are cached in plain form, there will be a section of the cache storing segments from common files and another storing segments from unique files, as shown in Fig.~\ref{fig:sys_model}. 

In the delivery phase, each user $k$ requests a single random file $d_k$ from the server. We denote the probability mass function (pmf) of the random request $d_k$ as $p_{d_k}^{[k]}$.

\begin{definition}
The $\text{demand set}$ for user $k$ is defined as $ S_k \triangleq \left\{  n \in [1, N_c+GN_u]: p_n^{[k]}>0 \right\}$, which represents the set of distinct files that can be requested by user $k$ with a positive probability. 
\end{definition}

It can be written that
\begin{equation}
    d_k \in S_k, \qquad \sum_{\forall{d\in S_k}} p_{d}^{[k]} = 1.
\end{equation}

\begin{definition}
The demand vector $\vec{d}=(d_1,\ldots,d_K)$ is defined as the set of files requested by users in the delivery phase, and $N(\vec{d})\in [1,\text{min} \{ K,N \}]$ denotes the number of distinct files in $\vec{d}$.
\end{definition}

Our goal will be to minimize the data rate (traffic from the server to the users) required to satisfy the users' requests. We consider two different metrics for such rate:

\begin{definition}
The peak and average rates of a coded caching scheme are respectively defined as
\begin{equation}
R^*(M) = \max_{\forall \vec{d}: d_k \in S_k} R_{\vec{d}}, \qquad \Bar{R} = \sum_{\forall \vec{d}, d_k \in S_k} p_{\vec{d}} R_{\vec{d}},
\end{equation}
where $R_{\vec{d}}$ denotes the number of bits transmitted to satisfy request vector $\vec{d}$.
\end{definition}

The average rate is highly dependent on the pmf of the requests $p_n^{[k]}$. In order to make the equations more tractable and facilitate the comparison with other coded caching schemes, our simulations will focus on the uniform-average rate, defined as follows.

\begin{definition}
The uniform-average-rate of a coded caching scheme is defined as
\begin{equation}
\Tilde{R} = \frac{1}{\prod_{k=1}^K |S_k|}\sum_{\forall \vec{d}, d_k \in S_k} R_{\vec{d}}.
\end{equation}
\end{definition}

This definition of uniform-average-rate is different from that in~\cite{yu2018characterizing}, where the distribution is simply uniform over $[N]^K$. It replaces the joint distribution $p_{\vec{d}}$ with a uniform distribution over the file demand set $S_1\times S_2\times\cdots\times S_K$.

\begin{figure}
\centering
\centerline{\includegraphics[width=.5\textwidth]{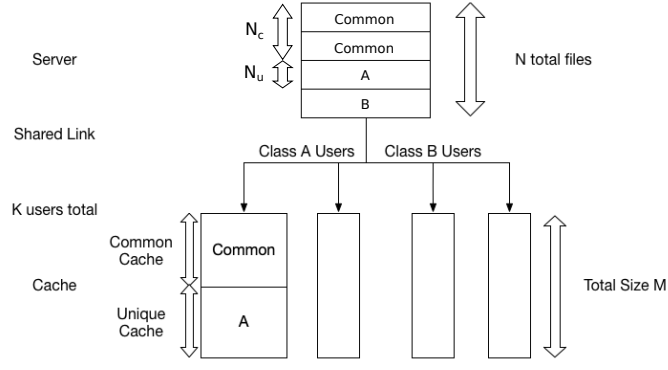}}
\caption{System Model with two distinct classes, A and B, each having two users. Each user's cache is divided into a section for common and another for unique files.}
\label{fig:sys_model}
\end{figure}

\subsection{Maddah-Ali and Niesen’s scheme}

This paper generalizes the centralized coded caching scheme with uncoded prefetching proposed by Maddah-Ali and Niesen~\cite{maddah2014fundamental}, from this point on referred to as MN's scheme, to heterogeneous user profiles. It is therefore important to review such scheme before we go any further.

In the placement phase, MN's scheme splits each file into $\binom{K}{t}$ non-overlapping segments, where $t = \frac{KM}{N}$. Each segment is cached by a distinct set of $t$ users, which results in each user caching $\binom{K-1}{t-1}$ segments per file. 
Specifically, this scheme is able to satisfy any vector of requests by transmitting at most $\binom{K}{t+1}$ messages of size $\binom{K}{t}^{-1}F$ bits. 
The peak rate (normalized by the file size $F$) is written as
\begin{align}\label{eq:MN}
R_{MN}(K,t) &=\frac{\binom{K}{t+1}}{\binom{K}{t}}\\
&= \frac{K-t}{t+1}\label{eq:MNsimplified}.
\end{align}
 If the server only receives requests for $m$ distinct files (e.g., only some of the users make a request, or their requests overlap), then the transmission rate with MN's scheme will become
\begin{equation}\label{eq:Rmreq}
R(K,\vec{d},t)=\frac{\binom{K}{t+1} - \binom{K-N(\vec{d})}{t+1}}{\binom{K}{t}},
 \end{equation}
 as was shown in~\cite{luo2019coded}.

 This paper will treat $t$ as continuous, just like it did with $K$, $G$, $N_c$, and $N_u$, to avoid integer effects. The next subsection explains how the combinatorial expressions in~Eq.~(\ref{eq:Rmreq}) can be extended to continuous arguments.

\subsection{Approximation of Transmission Rate for Simulation}
This subsection shows how Eq.~(\ref{eq:Rmreq}) can be extended into a continuous function over $0\leq t \leq K$.

When $t\leq 1$, the overall size of all the caches is not enough to store the $N$ files in full. It is therefore necessary to leave a fraction of each file out of the coded caching scheme and transmit it uncoded whenever that file is requested. The minimal such fraction can be found as $p=1-\frac{KM}{N}$.
Hence, according to~\cite{luo2018transfer}, the overall transmission rate for demand vector $\vec{d}$ when $t\leq1$ is
\begin{align}\label{eq:rate_t<1}
R &= N(\vec{d}) p + [\text{Rate if } t=1](1-p),
\end{align}
where $N(\vec{d})$ denotes the number of distinct files being requested.

When $K-1<t\leq K$, the opposite happens. The caches are large enough that a fraction of each file needs to be cached by every user, otherwise part of the caches would be left empty. It is therefore never necessary to transmit that fraction and coded caching schemes can be used to transmit the rest when the file is requested. The minimal such fraction can be found as $\gamma =K-t$. Hence, the overall transmission rate when $t\in (K-1,K]$ is
\begin{align}\label{eq:rate_t>K-1}
R &=0\cdot \gamma + [\text{Rate if } t=K-1](1-\gamma).
\end{align}

When $1\leq t\leq K-1$ but it is not an integer, Eq.~(\ref{eq:Rmreq}) is not well defined because the binomial coefficients require integer and strictly non-negative arguments. In order to interpolate these coefficients continuously, we use the Gamma function, which satisfies
$\Gamma(n) = (n-1)!$ for every integer $n$ and therefore
\begin{align}\label{eq:fact_gamma_relaxation}
\binom{n}{k} = \frac{n!}{k!(n-k)!} = \frac{\Gamma(n+1)}{\Gamma(k+1)\Gamma(n-k+1)}
\end{align}
without error when $n$ and $k$ are integers.

\section{Proposed Schemes}
\label{s-schemes}

The schemes proposed and analyzed in this paper are variations of MN's scheme and were first presented in our conference paper~\cite{wang2019coded}.

\subsection{Scheme 1: All common}
The system behaves as if all files are common during the placement phase, sacrificing local caching gain in favor of global caching gain. It ignores the distinction between all user profiles and requires every user to cache segments from every file, even if it would never request some of them. 
MN's scheme with $N=N_c+GN_u$ files is utilized for the placement and delivery. When $N\leq KM\leq(K-1)N$ the peak and average rates can be derived from Eqs.~(\ref{eq:MNsimplified})~and~(\ref{eq:Rmreq}), otherwise it becomes necessary to adjust their values as shown in Eqs.~(\ref{eq:rate_t<1})~and~(\ref{eq:rate_t>K-1}).

\begin{itemize}
    \item Peak rate: The peak rate is equivalent to that in MN's scheme with $N$ files and $K$ users. According to Eq.~(\ref{eq:MNsimplified}) it can be computed as:
    \begin{equation}
    \label{eq:R1peak}
    R_\mathrm{peak}^{(1)} = \frac{K-t_1}{t_1+1},
    \end{equation}
    where $t_1=\frac{KM}{N}$.

\item Average rate: The average rate is also equivalent to that in MN's scheme with $N$ files and $K$ users. Taking the expectation over the distribution of requests and using Eq.~(\ref{eq:Rmreq}) to compute the rate associated with each individual request vector yields:
\begin{equation}
    \label{eq:R1avg}
     R_\mathrm{avg}^{(1)} = \sum_{\forall{\vec{d}}} p_{\vec{d}} \frac{\binom{K}{t_1+1} - \binom{K- N_1(\vec{d})}{t_1+1}}{\binom{K}{t_1} } ,
    \end{equation}
where $p_{\vec{d}}$ denotes the probability associated to demand vector $\vec{d}$ and $N_1(\vec{d})$ denotes the number of distinct files requested by the $K$ users according to $\vec{d}$.

\end{itemize}

\subsection{Scheme 2: Split}
The system deals with common and unique files separately, decoupling their placement and delivery. A fraction $x$ of each user's cache is devoted to storing segments from common files and the remaining $(1-x)$ to store segments from unique files. 
Segments from common files are distributed over all $K$ users according to MN's scheme with $N_c$ files and $MFx$ bits of cache per user. Segments from unique files are only cached by the $K/G$ users in their corresponding class, also following MN's scheme to fill the remaining $MF(1-x)$ bits of cache capacity per user. The delivery phase is independent for common and unique files, never encoding segments from different file types in the same message.

If $\alpha$ out of the $K/G$ users in each class request unique files, the peak data rate for this scheme is given by
    \begin{align}
    R^{(2)}(x,\alpha) &= R_{c}(x,\alpha) + GR_{u}(x,\alpha)\label{eq:R2expand}\\
    &=\frac{\binom{K}{t_c+1}-\binom{G\alpha}{t_c+1}}{\binom{K}{t_c}}+G\frac{\binom{\frac{K}{G}}{t_u+1}-\binom{\frac{K}{G}-\alpha}{t_u+1}}{\binom{\frac{K}{G}}{t_u}},\nonumber
    \end{align}
    where $t_c = K\frac{Mx}{N_c}$ and $t_u = \frac{K}{G}\frac{M(1-x)}{N_u}$. The above expressions implicitly assume that $t_c$ and $t_u$ are both larger than 1 and smaller than $K-1$ and $\frac{K}{G}-1$, respectively. Otherwise, they need to be adjusted according to Eqs.~(\ref{eq:rate_t<1})~or~(\ref{eq:rate_t>K-1}).
     The two terms in Eq.~(\ref{eq:R2expand}) correspond to the rate required to deliver common files, $R_c(x,\alpha)$, and that required to deliver unique files, $R_u(x,\alpha)$, for each class. Despite users are only caching the files in their demand set, $x$ might favor unique files over common files (or vice versa), so the local caching gain is still being sacrificed for the benefit of global caching gain.

\begin{itemize}
    \item Peak rate: Theorem~\ref{thm:alphapeak} will later prove that, when the number of users is sufficiently large, the peak rate is achieved when the number of users requesting unique files is the same for every class. Hence, the peak rate can be found as
    \begin{equation}
    \label{eq:R2_peak}
        R_\mathrm{peak}^{(2)} = \min_x\max_\alpha R^{(2)}(x,\alpha),
    \end{equation}
    where the fraction $x$ is being optimized to minimize the peak rate.

\item Average rate: The average rate can be calculated as:
\begin{equation}\label{eq:Sch2eq}
R_\mathrm{avg}^{(2)} = R_c+\sum_{i=1}^G R_{u_i},
\end{equation}
consisting of the average transmission rate for common files $R_c$ and that for unique files $R_u$ in each class. Taking the expectation over the distribution of requests and using Eq.~(\ref{eq:Rmreq}) with a reduced memory $Mx$ and number of files $N_c$ to compute the rate associated with each individual request vector yields:
\begin{equation}
    \label{eq:Rc}
     R_\mathrm{c} = \sum_{\forall{\vec{d}}} p_{\vec{d}} \frac{\binom{K}{t_c+1} - \binom{K- N_c(\vec{d})}{t_c+1}}{\binom{K}{t_c} },
\end{equation}
where $t_c = \frac{K M x}{N_c}$ and $N_c(\vec{d})$ denotes the number of distinct common files requested by the $K$ users.
Similarly, the average transmission rate for unique files in the $i$-th class can be found by using Eq.~(\ref{eq:Rmreq}) with $K/G$ users, $N_u$ files, and capacity for $M(1-x)$ files in the cache:
\begin{equation}\label{eq:Ru}
R_{u_i} = \sum_{\forall{\vec{d}}} p_{\vec{d}} \frac{\binom{\frac{K}{G}}{t_u+1} - \binom{\frac{K}{G}- N_{u_i}(\vec{d})}{t_u+1}}{\binom{\frac{K}{G}}{t_u}},
\end{equation}
where $t_u = \frac{KM(1-x)}{G N_u}$ and $N_{u_i}(\vec{d})$ denotes the number of distinct unique files requested by the $K/G$ users in the $i$-th class.

\end{itemize}

\subsection{Scheme 3: All unique}
The system behaves as if all files are unique, maximizing local caching gain in detriment of global caching gain. It disregards the fact that common files can be requested by all user classes and independently applies MN's scheme for placement and delivery phases within each class of users. Instead of caching all $N_c+GN_u$ files, the users only cache the $N_c+N_u$ files corresponding to their class during the placement phase. When $G(N_c+N_u)\leq KM\leq(K-G)(N_c+N_u)$ the peak and average rates can be derived from Eqs.~(\ref{eq:MNsimplified})~and~(\ref{eq:Rmreq}), otherwise it becomes necessary to adjust their values as shown in Eqs.~(\ref{eq:rate_t<1})~and~(\ref{eq:rate_t>K-1}).

\begin{itemize}
    \item Peak rate: The peak rate for each class is equivalent to that in MN's scheme with $\frac{K}{G}$ users and $N_c+N_u$ files. Multiplying the rate in Eqs.~(\ref{eq:MNsimplified}) by the number of classes $G$ gives:
    \begin{equation}
    \label{eq:R3}
    R_\mathrm{peak}^{(3)} = G\frac{\frac{K}{G}-t_3}{t_3+1},
   \end{equation}
   where $t_3=\frac{K}{G}\frac{M}{N_c+N_u}$.

\item Average rate: The average rate for this scheme can be computed as the sum of the expected rates within each class. Using Eq.~(\ref{eq:Rmreq}) with $\frac{K}{G}$ users and $N_c+N_u$ files to compute the expected rates results in:
\begin{equation}
    \label{eq:R3_expand}
    R_\mathrm{avg}^{(3)} = \sum_{i=1}^G \sum_{\forall{\vec{d}}} p_{\vec{d}} \frac{\binom{\frac{K}{G}}{t_3+1} - \binom{\frac{K}{G}- N_{3_i}(\vec{d})}{t_3+1}}{\binom{\frac{K}{G}}{t_3} } ,
    \end{equation}
where $t_3 = \frac{K}{G}\frac{M}{N_c+N_u}$ 
and $N_{3_i}(\vec{d})$ represents the number of distinct files requested by the $K/G$ users in the $i$-th class.

\end{itemize}

\section{Lower Bound of Peak Rate}
\label{s-lower bound}
This section derives a lower bound for the peak rate of a coded caching system with heterogeneous user profiles and compares it with the peak rate of the three schemes from Section~\ref{s-schemes}.

\begin{theorem}\label{thm:thm_cut}
The peak rate of a coded caching scheme with $G$ user classes, $\frac{K}{G}$ users in each class, $N=N_c+GN_u$ total files, and cache size of $M$ files per user, can be bounded as
\begin{equation}\label{eq:Rpeak_bound}
    R^*(M)\geq \max_{s\in \left\{ 1,\ldots, \frac{\min(K,N)}{G}\right\}} \Big(  Gs - \frac{GsM}{ \left \lfloor \frac{ N}{Gs} \right \rfloor  }\Big).
\end{equation}
This result is based on a cut-set bound argument. 
\end{theorem}


\begin{proof}
 For simplicity, this proof assumes that $N_c$ is divisible by $G$. If that were not the case, we could just discard a few common files and prove the bound for a system with fewer files. Since the peak rate always increases with the size of the library, the bound would still hold for the original system.

Let $s \in \{1,2,\ldots,\frac{\text{min}\{K,N\}}{G}\}$ and consider the first $s$ users from each class $\gamma =1,2,\ldots,G$ denoting their caches $Z_1^{\gamma}, Z_2^{\gamma}, \ldots, Z_s^{\gamma}$. Divide the $N_c$ common files into $G$ sets so that each class has $\frac{N_c}{G}+N_u$ files associated with it. Denote them $\{W_1^{\gamma}, W_2^{\gamma},\ldots, W_{\frac{N}{G}}^{\gamma} \}$. Without loss of generality, we assume that the first $s$ files are requested for every class and the server fulfills those requests by transmitting $X_1$. The first $s$ users in each class must be able to recover $W_1^{\gamma}, W_2^{\gamma},\ldots, W_{s}^{\gamma}$ from their caches
$Z_1^{\gamma}, Z_2^{\gamma}, \ldots, Z_s^{\gamma}$ and $X_1$. Similarly, when the server sends $X_2$, the users in each class $\gamma$ are able to determine $W_{s+1}^{\gamma}, W_{s+2}^{\gamma},\ldots, W_{2s}^{\gamma}$ with their cache $Z_1^{\gamma}, Z_2^{\gamma}, \ldots, Z_s^{\gamma}, \gamma =1,2,\ldots,G$.
Continue in the same manner up to $X_{\lfloor \frac{N}{Gs} \rfloor }$.  We then have that $X_1, X_2, \ldots, X_{\lfloor \frac{N}{Gs}  \rfloor }$ and $Z_1^{\gamma}, Z_2^{\gamma},\ldots, Z_s^{\gamma}$ are enough to determine $W_{1}^{\gamma}, W_{2}^{\gamma},\ldots, W_{s\lfloor \frac{\lfloor\frac{N}{G}\rfloor}{s}  \rfloor}$, for $\gamma =1,2,\ldots,G$. Fig.~\ref{fig:cutsetbound} illustrates this setting.

\begin{figure}
\begin{center}
\includegraphics[width=.4\textwidth]{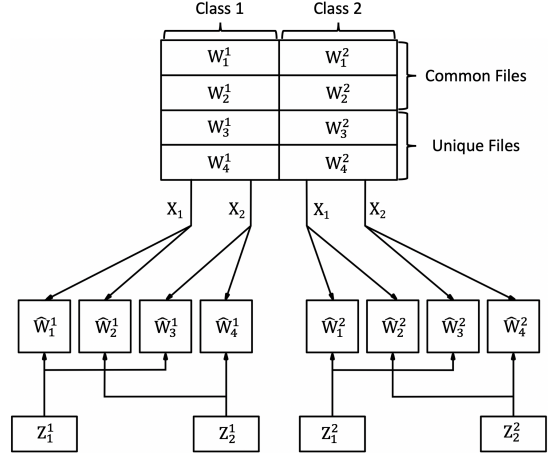}
\end{center}
\caption{Cut corresponding to parameter $s = 2$ in the proof of the converse. In the figure, $N_c =4, N_u=2, G=2, K=8.$}
\label{fig:cutsetbound}
\end{figure}

By the cut-set bound in~\cite{tm1991elements}, we can obtain that
\begin{equation}
    \left \lfloor \frac{N}{Gs} \right \rfloor R^*(M) + GsM \geq Gs\left \lfloor \frac{N}{Gs}  \right \rfloor .
\end{equation}
By solving for $R^*(M)$ and optimizing over all possible choices of $s$, it can be written that
\begin{equation}
    R^*(M)\geq \max_{s\in \left\{ 1,\ldots, \frac{\min(K,N)}{G} \right\} } \Big(  Gs - \frac{GsM}{ \left \lfloor \frac{ N}{Gs} \right \rfloor  }\Big),
\end{equation}
proving the theorem.
\end{proof}

We use $R_{CB}(s)$ to denote the argument maximized in the bound:
\begin{equation}\label{eq:CB}
    R_{CB}(s) = 
    Gs - \frac{GsM}{ \left \lfloor \frac{N}{Gs} \right \rfloor  }.
\end{equation}

\begin{remark}
When $M>\frac{N}{G}$ the above expression $R_{CB}(s)$ is negative for every $s$, and therefore the bound is trivial. Fortunately, $M\leq\frac{N}{G}$ in most practical cases, or every user would be able to cache most of the files in its demand set.
\end{remark}

\begin{remark}
When $G=1$, this bound reduces to the one derived in Theorem 2 of~\cite{maddah2014fundamental}.
\end{remark}




\begin{theorem}\label{thm:cut_set_factor1}
For the heterogeneous user profile model with $K$ users, a database
of $N = N_c+GN_u$ files, and a local cache size of $M$ files at each user with $\frac{N}{K}\leq M\leq\frac{N}{2G}$, 
it can be written that
\begin{equation}
    \frac{R_\mathrm{peak}^{(1)}}{R^*(M)} \leq 8.
\end{equation}
The restriction $\frac{N}{K}\leq M$ is imposed so that $t_1\geq 1$ and $M\leq\frac{N}{2G}$ to ensure that $R_{CB}(s)$ is not negative.
\end{theorem}

\begin{proof}
Loosening the bound in Eq.~(\ref{eq:Rpeak_bound}) results in
\begin{align}
R^*(M) &\geq \max_{s\in \left\{ 1,\ldots, \frac{\min(K,N)}{G} \right\} } R_{CB}(s)\\
&\geq \max_{0\leq s \leq \frac{\min(K,N)}{G}} \Big(  Gs - \frac{GsM}{ \frac{ N}{Gs}-1  }\Big)\\
&\geq \Big(  Gs_0 - \frac{Gs_0M}{ \frac{ N}{Gs_0}-1  }\Big),\label{eq:Rs_general}
\end{align}
where the last inequality holds for any $0\leq s_0\leq\frac{\min(K,N)}{G}$.

We first consider the case when $M\geq 1.5$. Let $s_0=\frac{N}{2GM}$ and observe that $0\leq s_0\leq \frac{\min(K,N)}{G}$ as long as $\max(1.5,\frac{N}{K})\leq M \leq\frac{N}{2G}$. Hence
\begin{align}
R^*(M) &\geq \Big(  Gs_0 - \frac{Gs_0M}{ \frac{ N}{Gs_0}-1  }\Big)\\
&\geq N\frac{M-1}{2M(2M-1)}.\label{eq:Rs0}
\end{align}
Dividing Eq.~(\ref{eq:R1peak}) by Eq.~(\ref{eq:Rs0}) and imposing $M\geq 1.5$ yields
\begin{align}
    \frac{R^{(1)}_\mathrm{peak}}{R^*(M)} &\leq \frac{N-M}{N}\cdot\frac{4KM}{KM+N}\cdot\frac{M-\frac{1}{2}}{M-1}\\
    &\leq 1\cdot 4 \cdot 2,
\end{align}
which proves that $R^{(1)}_\mathrm{peak}$ is within a factor of 8 from the optimal when $M\geq 1.5$.

Next consider the case when $M\leq 1.5$. Let $s_1=\frac{N}{5G}$ and assume for now that $0 \leq s_1\leq \frac{\min(K,N)}{G}$. Inserting this value in Eq.~(\ref{eq:Rs_general}) yields
\begin{align}\label{eq:Rs1}
R^*(M) &\geq \Big(  \frac{N}{5} - \frac{NM}{20}\Big)
\end{align}
and therefore
\begin{align}
    \frac{R^{(1)}_\mathrm{peak}}{R^*(M)} &\leq\frac{N-M}{N}\cdot \frac{K}{N+KM}\cdot\frac{20}{4-M}\\
    &\leq 1\cdot 1 \cdot 8,
\end{align}
which proves that $R^{(1)}_\mathrm{peak}$ is within a factor of 8 from the optimal when $M\leq 1.5$, as long as $s_1\leq \frac{\min(K,N)}{G}$.

It only remains to show that the theorem holds when $M\leq 1.5$ and $s_1\geq \frac{\min(K,N)}{G}$. This can only happen if $K\leq\frac{N}{5}$. Let $s_2=\frac{K}{G}$, for which Eq.~(\ref{eq:Rs_general}) becomes
\begin{align}\label{eq:Rs2}
R^*(M) &\geq K\frac{N-K(M+1)}{N-K}.
\end{align}
Therefore
\begin{align}
    \frac{R^{(1)}_\mathrm{peak}}{R^*(M)} &\leq\frac{N-M}{N+KM}\cdot \frac{N-K}{N-K(M+1)}\\
    &\leq 1\cdot \frac{0.8\ N}{0.5\ N}\\
    &\leq \frac{8}{5},
\end{align}
which proves that $R^{(1)}_\mathrm{peak}$ is within a factor of 1.6 from the optimal when $M\leq 2$ and $s_1\geq \frac{\min(K,N)}{G}$.
\end{proof}



\begin{theorem}\label{thm:cut_set_R_3}
For the heterogeneous user profile model with $K$ users, a database
of $N = N_c+GN_u$ files, and a local cache size of $M$ files at each user with $\frac{G}{K}(N_c+N_u)\leq M \leq \frac{N}{2G}$, it can be written that
\begin{equation}
  \frac{R_\mathrm{peak}^{(3)}}{R^*(M)} \leq 8\frac{K}{G}.
\end{equation}
The restriction $\frac{G}{K}(N_c+N_u)\leq M$ is imposed so that $t_3\geq 1$ and $M \leq \frac{N}{2G}$ to ensure that $R_{CB}(s)$ is not negative.
\end{theorem}

\begin{proof}
For $M\geq 1.5$, Eq.~(\ref{eq:Rs0}) in the proof of Theorem~\ref{thm:cut_set_factor1} showed that
\begin{align}
R^*(M) &\geq N\frac{M-1}{2M(2M-1)}.
\end{align}
Dividing Eq.~(\ref{eq:R3}) by the above expression gives
\begin{align}
    \frac{R^{(3)}_\mathrm{peak}}{R^*(M)} &\leq \frac{2KM}{N}\cdot\frac{G(N_c+N_u-M)}{G(N_c+N_u)+KM}\cdot\frac{2M-1}{M-1}\\
    &\leq \frac{K}{G}\cdot 1 \cdot 4.
\end{align}

For $M\leq 1.5$, Eq.~(\ref{eq:Rs1}) in the proof of Theorem~\ref{thm:cut_set_factor1} showed that
\begin{align}
R^*(M) &\geq \Big(  \frac{N}{5} - \frac{NM}{20}\Big)
\end{align}
as long as $K\geq\frac{N}{5}$. Dividing Eq.~(\ref{eq:R3}) by the above expression gives
\begin{align}
    \frac{R^{(3)}_\mathrm{peak}}{R^*(M)}
    &\leq \frac{K}{N}\cdot\frac{G(N_c+N_u-M)}{G(N_c+N_u)+KM}\cdot\frac{20}{4-M}\\
    &\leq \frac{K}{N}\cdot 1\cdot 8\\
    &\leq 8\frac{K}{G}.
\end{align}

Finally, for $M\leq 1.5$ and $K\leq\frac{N}{5}$ we can divide Eq.~(\ref{eq:R3}) by Eq.~(\ref{eq:Rs2}) to obtain
\begin{align}
    \frac{R^{(3)}_\mathrm{peak}}{R^*(M)}
    &\leq \frac{N-K}{N-K(M+1)}\cdot\frac{G(N_c+N_u-M)}{G(N_c+N_u)+KM}\\
    &\leq \frac{N-K}{N-K(2.5)}\cdot 1\\
    &\leq \frac{8}{5}.
\end{align}
There must be more users $K$ than classes $G$, so $\frac{K}{G}\geq1$ and the theorem is proved.
%
%
%
%
%
%
\end{proof}

\begin{theorem}\label{thm:cut_set_factor2}
For the heterogeneous user profile model with $K$ users, a database
of $N = N_c+GN_u$ files, and a local cache size of $M$ files at each user with $\frac{G}{K}(N_c+N_u)\leq M \leq \frac{N}{2G}$, it can be written that
\begin{equation}\label{eq:theorem4}
    \frac{R_\mathrm{peak}^{(2)}}{R^*(M)} < 8+8\frac{K}{G}.
\end{equation}
The restriction $\frac{G}{K}(N_c+N_u)\leq M$ is imposed so that $t_3\geq 1$ and $M \leq \frac{N}{2G}$ to ensure that $R_{CB}(s)$ is not negative.
\end{theorem}
\begin{proof}
The peak rate of Scheme 2 can be bound as follows:
\begin{align}
    R_\mathrm{peak}^{(2)} &= \min_x\max_\alpha R^{(2)}(x,\alpha)\\
    &\leq \max_\alpha R^{(2)}\left(x=\frac{N_c}{N},\alpha\right)\\
    &\leq R_c\left(x=\frac{N_c}{N},\alpha=0\right)+R_u\left(x=\frac{N_c}{N},\alpha=\frac{K}{G}\right)\\
    &\leq R_\mathrm{peak}^{(1)}+G\frac{\frac{K}{G}-t_1}{t_1+1},
\end{align}
where $N=N_c+GN_u$ and $t_1=\frac{KM}{N}$. Since $t_3\leq t_1$ and $R_\mathrm{peak}^{(3)}$ decreases monotonically with $t_3$ we can conclude that
\begin{align}
    R_\mathrm{peak}^{(2)} &\leq R_\mathrm{peak}^{(1)} + R_\mathrm{peak}^{(3)}.
\end{align}
Finally, we apply Theorems~\ref{thm:cut_set_factor1}~and~\ref{thm:cut_set_R_3} to obtain Eq.~(\ref{eq:theorem4}).
\end{proof}

We do not attempt to characterize a bound for average rate because it would depend on the popularity distribution of the files. A bound for uniform-average rate could be derived, but we believe that it would not provide valuable insights for the general case.

\section{Results}
\label{s-Results}

This first part of this section studies how to optimize the distribution of cache between common and unique files in Scheme 2 so that the peak rate and uniform-average rate are minimized. We discover that when users' cache storage is small, devoting all the cache to common files will minimize both the peak and the average rate of transmission. The second part of the section provides detailed comparisons between the peak rates of the three schemes proposed in Section~\ref{s-schemes} and analyzes which scheme offers the best performance for each value of $M$. A partial summary of results can be found in Table~\ref{tab:my_label}.

 Our previous conference paper~\cite{zhang202averagerate} provided some partial and asymptotic results for uniform-average rate, but we have decided not to include those here, postponing them to future work on a separate paper.

\subsection{Optimizing $x$ for Scheme 2}

When $M$ is large, users are able to cache most of the files and the choice of $x$ is less relevant. Furthermore, scenarios where caches are almost as large as the whole library rarely come up in practical applications. Hence, we will focus our analysis on the case with relatively small $M$ compared with the size of the library $N$.

\begin{theorem}\label{pr:schemeII_peak}
When $M \leq\frac{N_c}{K} \left[ \sqrt{\frac{N_u(K+1)}{N_c(\frac{K}{G}+1)}}-1 \right]$, the peak rate of Scheme~2 is minimized by devoting all the cache to common files ($x=1$).
\end{theorem}
\begin{proof}
See Appendix.
\end{proof}


\begin{theorem}\label{pr:propI}
When the number of users $K$ is large and $M\leq\frac{1}{K}\min(N_c,GN_u)$, the uniform-average rate of Scheme~2 is minimized by devoting all the cache to either common or unique files. Specifically, it should all be devoted to common files ($x=1$) when
\begin{equation}\label{eq:avg_devotecommon}
\frac{N_u}{N_c}> G\cdot\frac{\mathbb{E}^2[N_u(\vec{d})] + \mathbb{E}[N_u(\vec{d})] }{ \mathbb{E}^2[N_c(\vec{d})]+\mathbb{E}[N_c(\vec{d})] },
\end{equation}
otherwise it should all be devoted to unique files ($x=0$). In Eq.~(\ref{eq:avg_devotecommon}),
\begin{align}
\mathbb{E}[N_c(\vec{d})] &= N_c\left[ 1- \left( \frac{N_c+N_u-1}{N_c+N_u} \right)^K \right],\\
\mathbb{E}[N_{u}(\vec{d})] &= N_u\left[ 1- \left( \frac{N_c+N_u-1}{N_c+N_u} \right) ^{K/G} \right].
\end{align}
\end{theorem}

\begin{proof}
See Appendix.
\end{proof}
Theorem~\ref{pr:propI} generalizes Prop. 3 from paper~\cite{chang2020twousers}, where there exist two classes with one user each. When the number of common files is large and the cache size is below half of the common files, the users should only cache common files. 





\begin{corollary}
If the users' devices have relatively small storage 
and the number of common files is not too large,
it is recommended for the users to devote all their cache to common files and transmit the unique files uncoded.
\end{corollary}


%
%


\subsection{Peak Rate Comparison}

First, we compare the peak rate of Schemes 1 and 3, since they have the simplest expressions.

\begin{theorem}\label{thm:common_unique}
When $M$ is large enough, Scheme 3 offers lower peak rate than Scheme 1, and vice versa. Specifically,
\begin{equation}
R_\mathrm{peak}^{(3)}\leq R_\mathrm{peak}^{(1)}\qquad \Leftrightarrow \qquad M\geq N_c-\frac{GN_u}{K}.
\end{equation}
\end{theorem}

\begin{proof}
This theorem can be proved by simple manipulation of Eqs~(\ref{eq:R1peak})~and~(\ref{eq:R3}).
\end{proof}

\begin{corollary}
When $M$ is small, it is often beneficial for users to cache segments from undesired files, to increase multicasting opportunities. The loss in local caching gain is more than compensated by the gain in global caching gain~\cite{chang2019coded}.
\end{corollary}

In Schemes 1 and 3, the number of users requesting common versus unique files is irrelevant, since segments from both files can be encoded together. In Scheme 2, however, it plays a major role. We now intend to show that in order to compute the peak rate, we only need to consider the case where the subdivision is the same for all user classes.



\begin{theorem}
\label{thm:alphapeak}
There exists a number $\alpha\in(0,\frac{K}{G})$ such that the peak rate for Scheme 2 is achieved when every class has $\alpha$ users requesting unique files.
\end{theorem}

\begin{proof}
Let $\alpha_i$ represent the number of users from class $i$ requesting unique files, and assume that $\mathbf{\alpha}=(\alpha_1,\ldots, \alpha_G)$ maximizes the rate, given by
\begin{equation}
R(\mathbf{\alpha})=R_c\left(\frac{1}{G}\sum^{G}_{i=1}\alpha_i\right) + \sum^{G}_{i=1}R_u(\alpha_i),
\end{equation}
where $R_c$ and $R_u$ have been defined by Eq.~(\ref{eq:R2expand}) and we omit $x$ for simplicity.

Without loss of generality, assume $\alpha_1>\alpha_2$ and let $\mathbf{\beta}=(\alpha_1-1,\alpha_2+1,\alpha_3,\ldots,\alpha_G)$. Then
\begin{align}
R(\mathbf{\beta})&-R(\mathbf{\alpha})=   R_u(\alpha_1-1)-R_u(\alpha_1)+R_u(\alpha_2+1)-R_u(\alpha_2).
\end{align}
We now prove that
\begin{equation}
\label{eq:compare_R_u}
    R_u(\alpha_2+1)-R_u(\alpha_2)\geq R_u(\alpha_1)-R_u(\alpha_1-1).
\end{equation}
This result follows from the fact that the rate is submodular in the number of requests, but we prove it anyway. With $\alpha_1>\alpha_2$, Eq.~(\ref{eq:compare_R_u}) can be written as
\begin{align}
\frac{\binom{\frac{K}{G} - \alpha_2}{t_u+1} - \binom{\frac{K}{G}-\alpha_2-1}{t_u+1}}{\binom{\frac{K}{G}}{t_u}} &\geq \frac{\binom{\frac{K}{G}-\alpha_1+1}{t_u+1} - \binom{\frac{K}{G}-\alpha_1}{t_u+1}}{\binom{\frac{K}{G}}{t_u}}\\
\binom{\frac{K}{G} - (\alpha_2 + 1)}{t_u} &\geq \binom{\frac{K}{G}-\alpha_1}{t_u},
\end{align}
which is true, since binomial coefficients increase monotonically with the number of elements.

Therefore, $\mathbf{\beta}$ achieves a rate at least as high as $\mathbf{\alpha}$ with less variance across the coefficients. For large enough $K$ (\ie using the continuous relaxation of the problem), we can conclude that a uniform set of coefficients would achieve the peak rate.
%
%
%
\end{proof}

We are now ready to compare the peak rate of Scheme 2 with that for the other two.

\begin{theorem}\label{thm:split_unique}
When $M$ is large enough, Scheme 2 offers lower peak rate than Scheme 3. Specifically,
\begin{equation}\label{eq:R2R3threshold}
R_\mathrm{peak}^{(2)}\leq R_\mathrm{peak}^{(3)}\qquad \Leftarrow \qquad M\geq \frac{G}{G-1}\frac{K+1}{K}N_u.
\end{equation}
\end{theorem}

\begin{proof}
If $x=1-\frac{N_u}{M}$, each user stores all the unique files that it might request. The worst case $\alpha$ is therefore $\alpha=0$. Observe that
\begin{align} \min_{x}\max_\alpha R^{(2)}(x,\alpha) &\leq \max_\alpha R_\mathrm{peak}^{(2)}\left(1-\frac{N_u}{M},\alpha\right)\\
&= R_\mathrm{peak}^{(2)}(1-\frac{N_u}{M},0)\\
&= K\frac{N_u+N_c-M}{K(M-N_u)+N_c}. \label{eq:R2upperbound}
\end{align}
After some rearrangement, Eq.~(\ref{eq:R3}) can be written as
\begin{equation}
R_\mathrm{peak}^{(3)}=KG\frac{N_c+N_u-M}{KM+G(N_c+N_u)}.
\end{equation}
 A simple comparison of the last two equations yields Eq.~(\ref{eq:R2R3threshold}).
\end{proof}

\begin{theorem}\label{thm:common_split}
When $M$ is large enough, Scheme 2 provides lower peak rate than Scheme 1. Specifically,
\begin{equation}
R_\mathrm{peak}^{(2)}\leq R_\mathrm{peak}^{(1)}\qquad \Leftarrow \qquad M\geq \frac{N_c}{G}+N_u. \label{eq:R2R1compare}
\end{equation}
\end{theorem}
\begin{proof}
From Eq.~(\ref{eq:R2upperbound}) we can observe that
\begin{equation}
    R_\mathrm{peak}^{(2)} \leq K\frac{N_u+N_c-M}{K(M-N_u)+N_c},
\end{equation}
and Eq.~(\ref{eq:R1peak}) can be rearranged as
\begin{equation}
    R_\mathrm{peak}^{(1)} = \frac{K-\frac{KM}{N_c+GN_u}}{\frac{KM}{N_c+GN_u}+1}.
\end{equation}
By comparing these two equations, we are able to generate Eq.~(\ref{eq:R2R1compare}).
\end{proof}

\begin{corollary}
Scheme 2 provides the lowest peak rate of the three when $M\geq\max\left(\frac{G}{G-1}\frac{K+1}{K}N_u,\frac{N_c}{G}+N_u\right)$.
\end{corollary}

\begin{proof}
This corollary can be simply proven by combining Theorem~\ref{thm:split_unique} and Theorem~\ref{thm:common_split}, setting $M$ to be the larger value between the two.
%
\end{proof}



\begin{theorem}
When $M$ is small enough, Scheme 1 offers lower peak rate than Scheme 2. Specifically,
\begin{equation}\label{eq:R1R2threshold}
R_\mathrm{peak}^{(1)}\leq R_\mathrm{peak}^{(2)}\qquad \Leftarrow \qquad M\leq \frac{\min(N_c,GN_u)}{K}.
\end{equation}
\end{theorem}

\begin{proof}
If $M\leq \frac{\min(N_c,GN_u)}{K}$ then $t_1\leq 1$ and we can combine Eqs.~(\ref{eq:R1peak})~and~(\ref{eq:rate_t<1}) to obtain
\begin{align}
     R_\mathrm{peak}^{(1)} &=  K-\frac{M}{2}\frac{K(K+1)}{N_c+GN_u}\\
     &\leq K-\gamma\frac{K+1}{4},\label{eq:R1peakgamma}
\end{align}
where $\gamma=\frac{KM}{\max(N_c,GN_u)}$.

As for $R_\mathrm{peak}^{(2)}$, it is defined as the highest rate experienced for any number of unique requests $\alpha$:
\begin{align}
    R_\mathrm{peak}^{(2)} &= \min_{x}\max_{\alpha} R^{(2)}\left(x,\alpha\right).
\end{align}
If $M\leq \frac{\min(N_c,GN_u)}{K}$ then $t_c\leq x$ and $t_u\leq 1-x$. Combining Eqs.~(\ref{eq:R2expand})~and~(\ref{eq:rate_t<1}) then yields
\begin{align}
    R^{(2)}\left(x,\alpha\right) &=\frac{\binom{K}{2}-\binom{G\alpha}{2}}{\binom{K}{1}} t_c+(K-G\alpha)(1-t_c)
     + G\frac{\binom{\frac{K}{G}}{2}-\binom{\frac{K}{G}-\alpha}{2}}{\binom{\frac{K}{G}}{1}} t_u+G\alpha(1-t_u),
\end{align}
where $t_c=\frac{KMx}{N_c}$ and $t_u=\frac{KM(1-x)}{GN_u}$. It can be shown that $R^{(2)}\left(x,\alpha\right)$ is monotonically increasing with $N_c$ and $N_u$ or, equivalently, monotonically decreasing with $t_c$ and $t_u$. Replacing $t_c$ and $t_u$ with $\gamma x$ and $\gamma(1-x)$, respectively, should therefore reduce the value of $R^{(2)}(x,\alpha)$:
\begin{align}
    R^{(2)}\left(x,\alpha\right) &\geq K-\gamma\left(\frac{K + 1}{2}x\right.
    \left.+\frac{G\alpha[G\alpha + G - (2K + G+ 1)x]}{2K}\right).
\end{align}
This is a quadratic equation with respect to $\alpha$, which attains its maximum value at $\alpha^\star=\frac{Gx + 2Kx - G + x}{2G}$. Similarly, $R^{(2)}\left(x,\alpha^\star\right)$ is a quadratic equation of $x$ which can be minimized to obtain
\begin{align}
    R^{(2)}_\mathrm{peak}&=\min_x\max_\alpha R^{(2)}(x,\alpha)\nonumber\\
    &\geq \min_xR^{(2)}(x,\alpha^\star)\\
    &\geq K-\gamma\frac{K+1}{2}\left(\frac{(G + K+1)(G+K)}{(G + 2K + 1)^2}\right)\label{eq:GNugeqNc1}\\
    &\geq K-\gamma\frac{K+1}{2}\left(\frac{1}{2}-\frac{K^2+(K+1)^2-G^2}{2(G + 2K + 1)^2}\right)\\
    &\geq R^{(1)}_\mathrm{peak},
\end{align}
where the last inequality results from comparing Eq.~(\ref{eq:GNugeqNc1}) with Eq.~(\ref{eq:R1peakgamma}) while keeping in mind that $G\leq K$.
\end{proof}

\begin{corollary}\label{cor:Mcompare}
For small enough $M$, Scheme 1 yields lower peak rate than Schemes 2 and 3. For large enough $M$, Scheme 2 offers the lowest rate of the three. In some cases, there is a range of intermediate $M$ values for which Scheme 3 has lower rate than the other two.
\end{corollary}

\begin{table}[]
    \centering
    \begin{tabular}{|c|c|}
    \hline
     $M$ range & Result \\
     \hline
       $M\leq \min \left(N_c-\frac{GN_u}{K}, \frac{N_c}{K}, \frac{GN_u}{K} \right)$  & $R_\mathrm{peak}^{(1)}$ best\\
      $M\geq\max\left(\frac{G}{G-1}\frac{K+1}{K}N_u,\frac{N_c}{G}+N_u\right)$   & $R_\mathrm{peak}^{(2)}$ best\\
       $\frac{N}{K} \leq M \leq \frac{N}{2G}$  & $R_\mathrm{peak}^{(1)}/R^*(M)\leq 8$\\
       $\frac{G}{K}(N_c+N_u) \leq M \leq \frac{N}{2G}$  & $R_\mathrm{peak}^{(2)}/R^*(M) < 8+8\frac{K}{G}$\\
       $\frac{G}{K}(N_c+N_u) \leq M \leq \frac{N}{2G}$  & $R_\mathrm{peak}^{(3)}/R^*(M)\leq 8\frac{K}{G}$\\
         \hline
    \end{tabular}
    \caption{Summary of peak rate results}
    \label{tab:my_label}
\end{table}

\section{Numerical Simulations}
\label{s-simulations}

This section provides simulations illustrating the peak and uniform-average rates of the three proposed schemes, as well as lower bounds to provide a framework for comparison. To the extent of our knowledge, there are no schemes in the literature which could be suitable for our scenario with heterogeneous user profiles. The best performing schemes for homogeneous users are those found by solving combinatorial optimization problems as described in in~\cite{jin2017structural}~and~\cite{daniel2019optimization}. However, with uniform file popularities, cache capacities and file sizes, those schemes are equivalent to Maddah-Ali and Niesen's scheme. Solving the optimization problems while ignoring user classes results in Scheme~1 and doing it independently for each class results in Scheme~3.



Fig.~\ref{fig:cur_set_bound} illustrates the peak rate of the three schemes and the cut set bound, for different cache sizes $M$ and number of classes $G$. The number of users per class is set as $K/G=8$ in all cases. The results match the statement in Corollary~\ref{cor:Mcompare}: it is better to use Scheme 1 (all common) for small cache sizes and Scheme 2 (split) for large ones, regardless of the number of classes. This result seems counter-intuitive, since it suggests that every user should cache segments from every file when the caches are small, even as the number of unique files scales with the number of classes. However, it turns out that the multicasting gain more than compensates for the loss in local caching.

The peak rate values increase with the number of classes due mainly to the increase in the number of files and users. It can be seen that the peak rate of Scheme~3 (all unique) increases above the others, reaching a point when it is never the preferred option. Again, this is somewhat counter-intuitive; it seems like a good idea to deal with each class independently when the number of classes is large, but it is not.

\begin{figure}
\begin{center}
\centerline{\includegraphics[width=.4\textwidth]{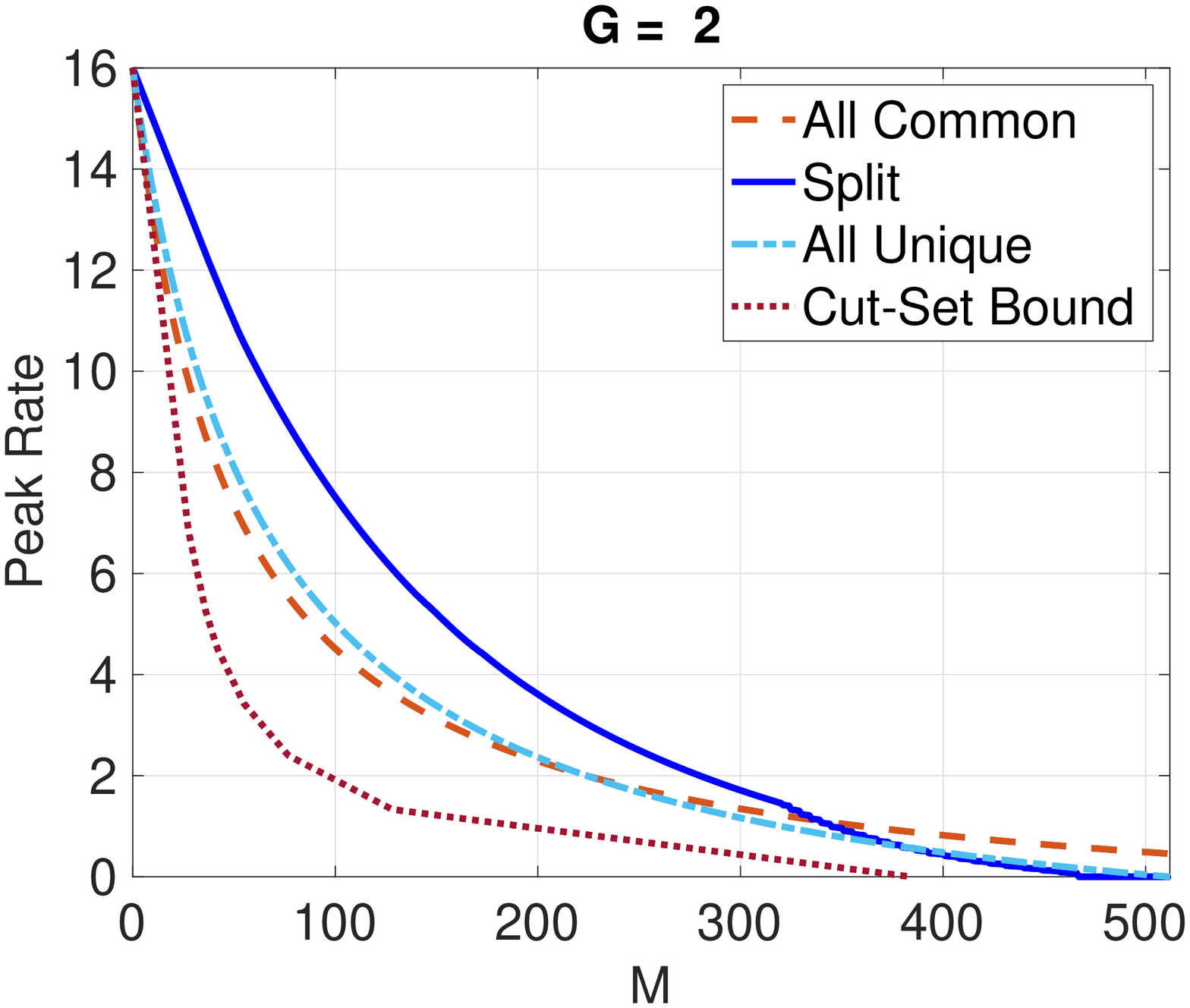}}
\centerline{\includegraphics[width=.4\textwidth]{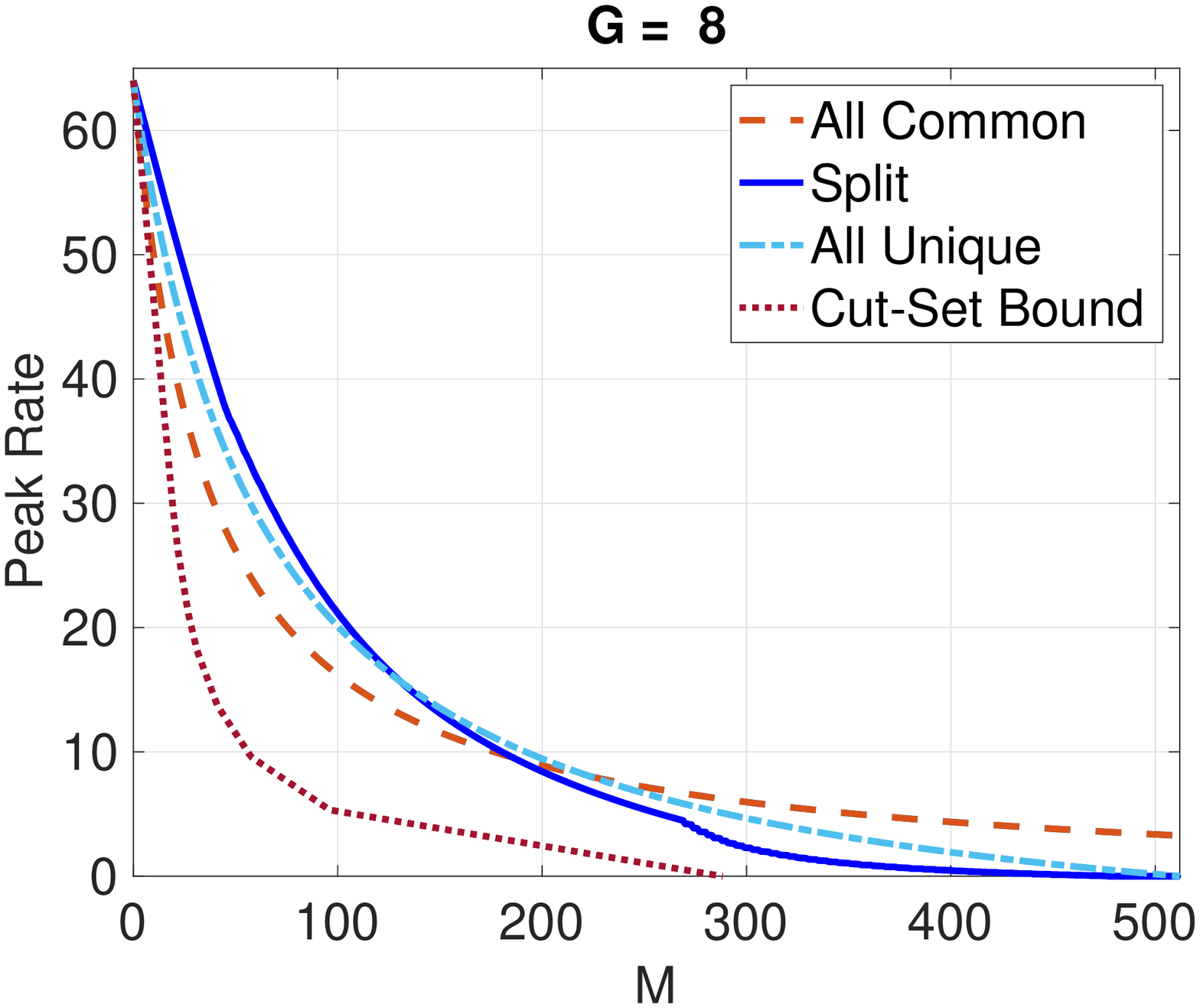}}
\centerline{\includegraphics[width=.4\textwidth]{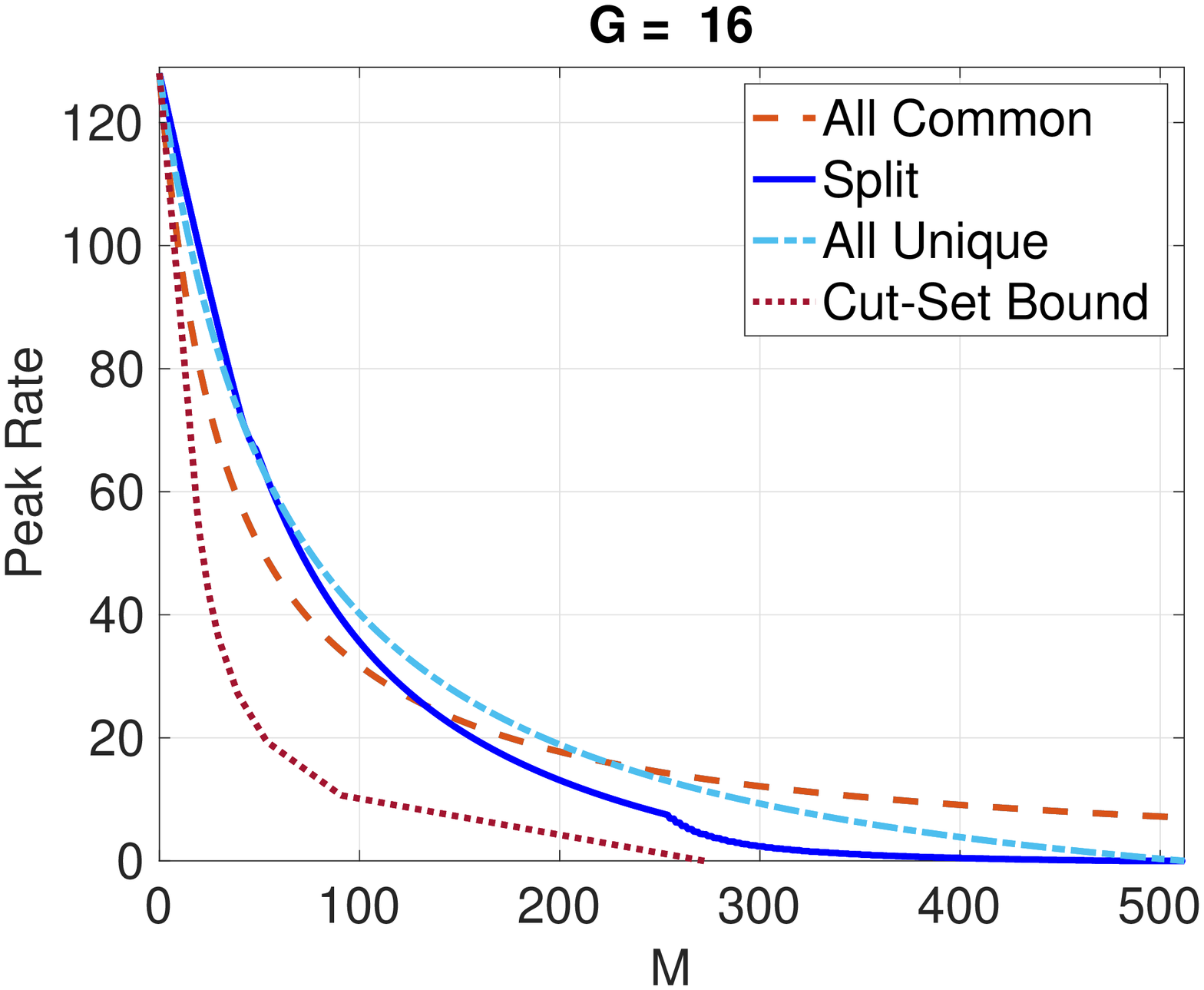}}
\end{center}
\caption{Peak rates and cut-set bound vs cache size ($M$) for $N_c=256$, $N_u=256$, and 8 users per class.}
\label{fig:cur_set_bound}
\end{figure}

\begin{figure}
\begin{center}
\centerline{\includegraphics[width=.4\textwidth]{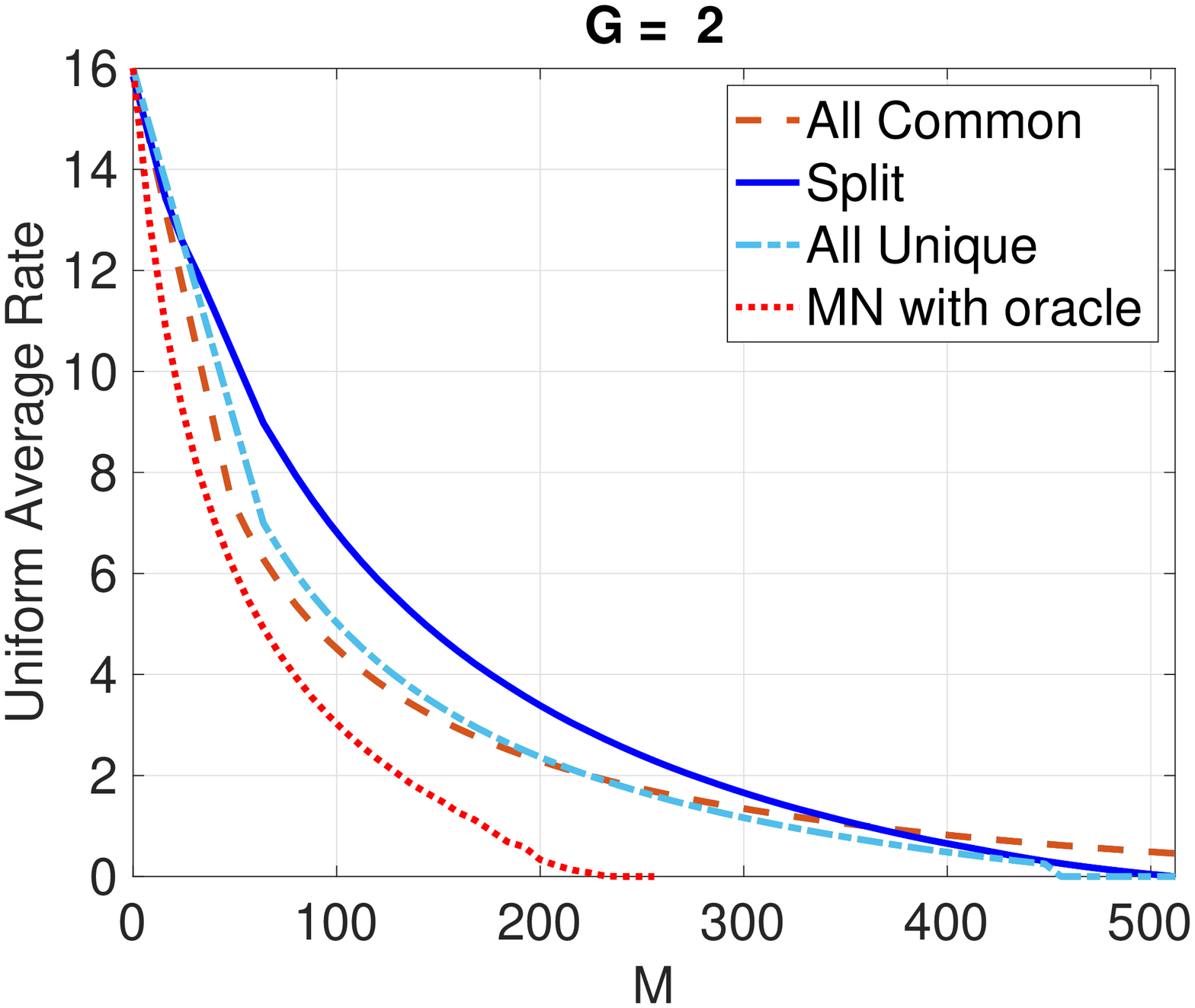}}
\centerline{\includegraphics[width=.4\textwidth]{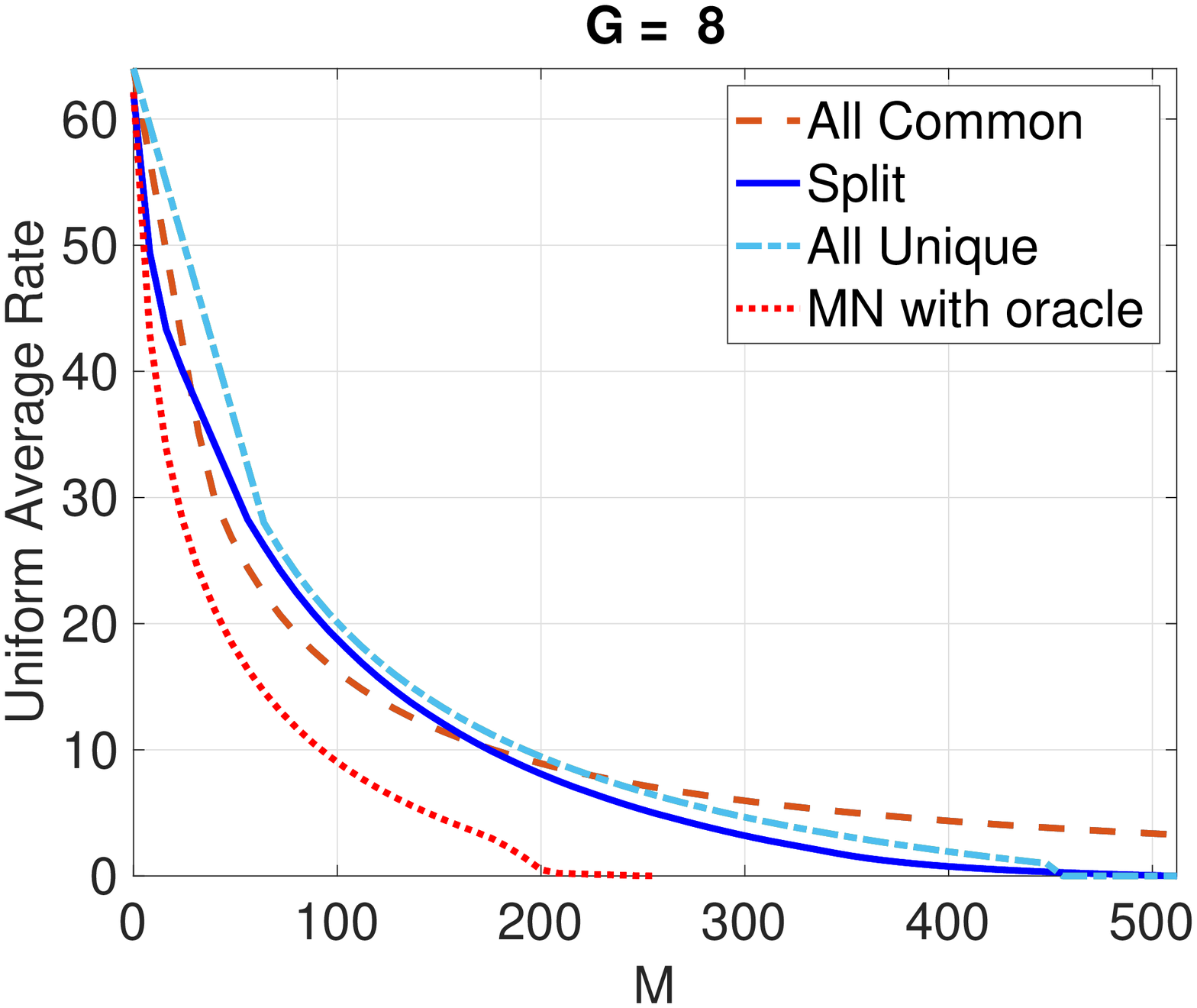}}
\centerline{\includegraphics[width=.4\textwidth]{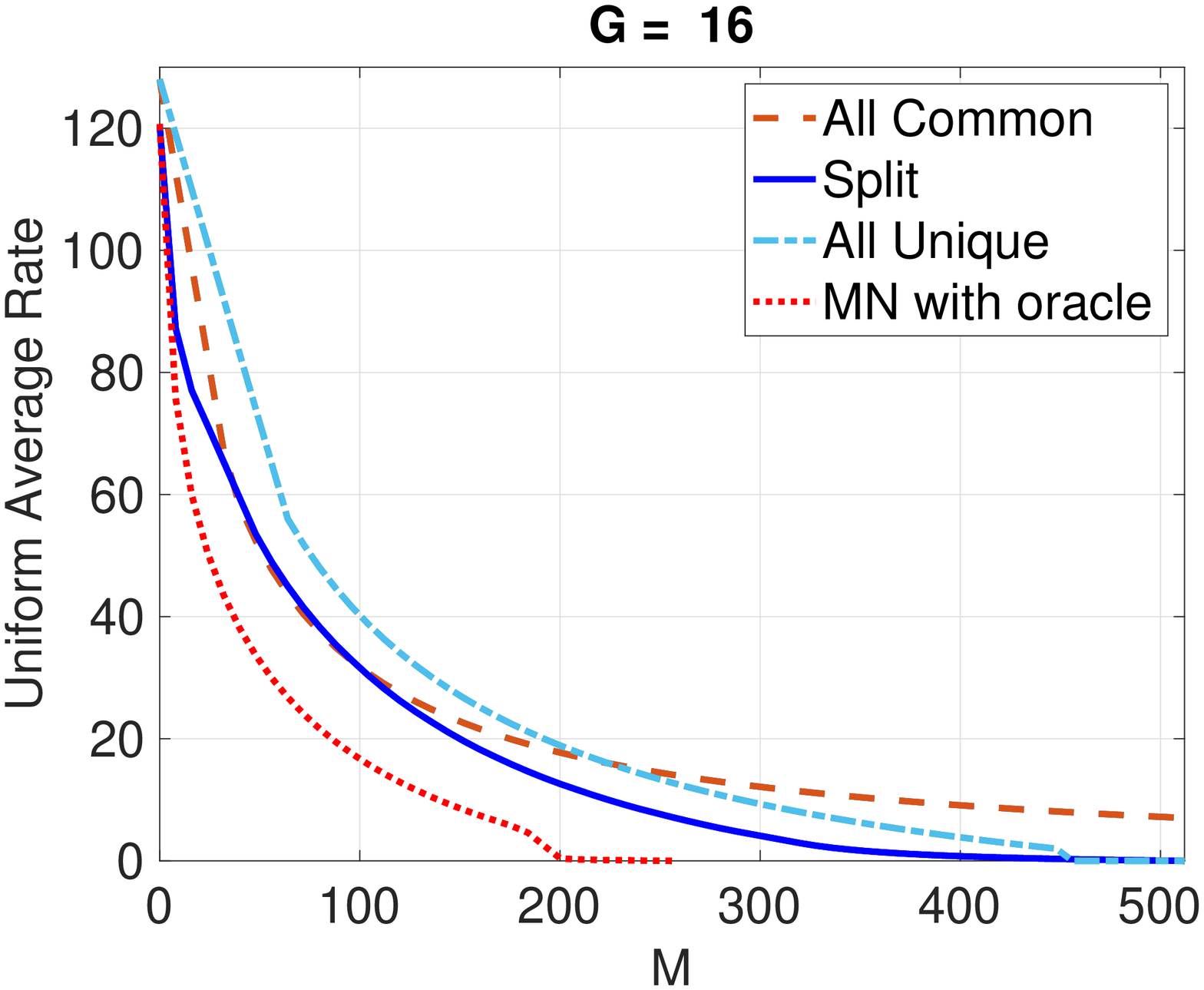}}
\end{center}
\caption{Average rate vs Cache size ($M$) for $N_c=256$, $N_u=256$, and 8 users per class.}
\label{fig:ave_oracle}
\end{figure}

Fig.~\ref{fig:ave_oracle} shows the uniform-average rates of the three schemes in the same scenario. The minimal uniform-average rate in scenarios with heterogeneous user profiles is unknown, so we define a new scheme "MN with oracle" to provide an approximate lower bound. In this scheme the system knows in advance which users will request common and unique files, and it populates their caches using MN's scheme for common and unique files separately. This results in the following uniform-average rate:
\begin{align}
    \label{eq:R_oracle}
     R_\mathrm{orc} = \sum_{k_c=0}^{K} p_{k_c} \Bigg( &E\left[R(k_c,m,t_{oc})\right]
     +G\cdot E\left[R\left(\frac{K-k_c}{G},m,t_{ou}\right)\right]\Bigg),\nonumber
\end{align}
where $k_c$ represents the number of users that request common files, $E[R(K,m,t)]$ is the expectation of the rate defined in Eq.~(\ref{eq:Rmreq}) over the number of distinct files requested $m$, $t_{oc} = \frac{k_cM}{N_c}$, $t_{ou} = \frac{(K-k_c)M}{GN_u}$, 
and
\begin{equation}
    \label{eq:p_binomial}
     p_{k_c} = \dbinom{K}{k_c}\Big(\frac{N_c}{N_c+N_u}\Big)^{k_c} \Big(\frac{N_u}{N_c+N_u}\Big)^{K-k_c}
\end{equation}
is the probability that $k_c$ users request common files.

The results suggest that Scheme 2, which splits the placement and delivery of common and unique files, is highly suboptimal when the number of classes is small, unless the cache memory is very small or very large. However, when the number of classes increases, Scheme 2 achieves the lowest average rate among all three of the schemes. 
This result aligns with Prop.~3 in \cite{zhang202averagerate}.

It is worth noting that these results are different from those previously observed for the peak rate: for small $M$, Scheme 2 presents the lowest uniform-average rate and the highest peak rate among the three schemes, regardless of the number of classes.


Fig.~\ref{fig:avg_vs_class} investigates the performance of the three schemes as the number of classes grows. In this scenario, $N_c = N_u =256$, there are $8$ users in each class and we set $M =256$ to provide enough storage for each user to cache half of the files it could request. The top plot stands for the comparison of the peak rates of three schemes, and the bottom plot compares the average rates of three schemes. As the number of classes increases, so does the number of total users. This results in the degeneration of the performance of three schemes. 
Schemes 1 (all common) and 3 (all unique) suffer nearly linear degradation, while Scheme 2 (split) scales better. This is because Scheme 2 (split) is able to adjust its cache distribution as the number of classes increases so that both peak and average rate are minimized. Therefore, when there are more classes of users joining the computer communication network, Scheme 2 (split) is able to provide better video streaming service than other schemes and its performance is scarcely effected. Similar results for $N_c=64$, $N_u =64$ and $M=64$ can be found in papers~\cite{wang2019coded}~and~\cite{zhang202averagerate}.

\begin{figure}
\begin{center}
\centerline{\includegraphics[width=.4\textwidth]{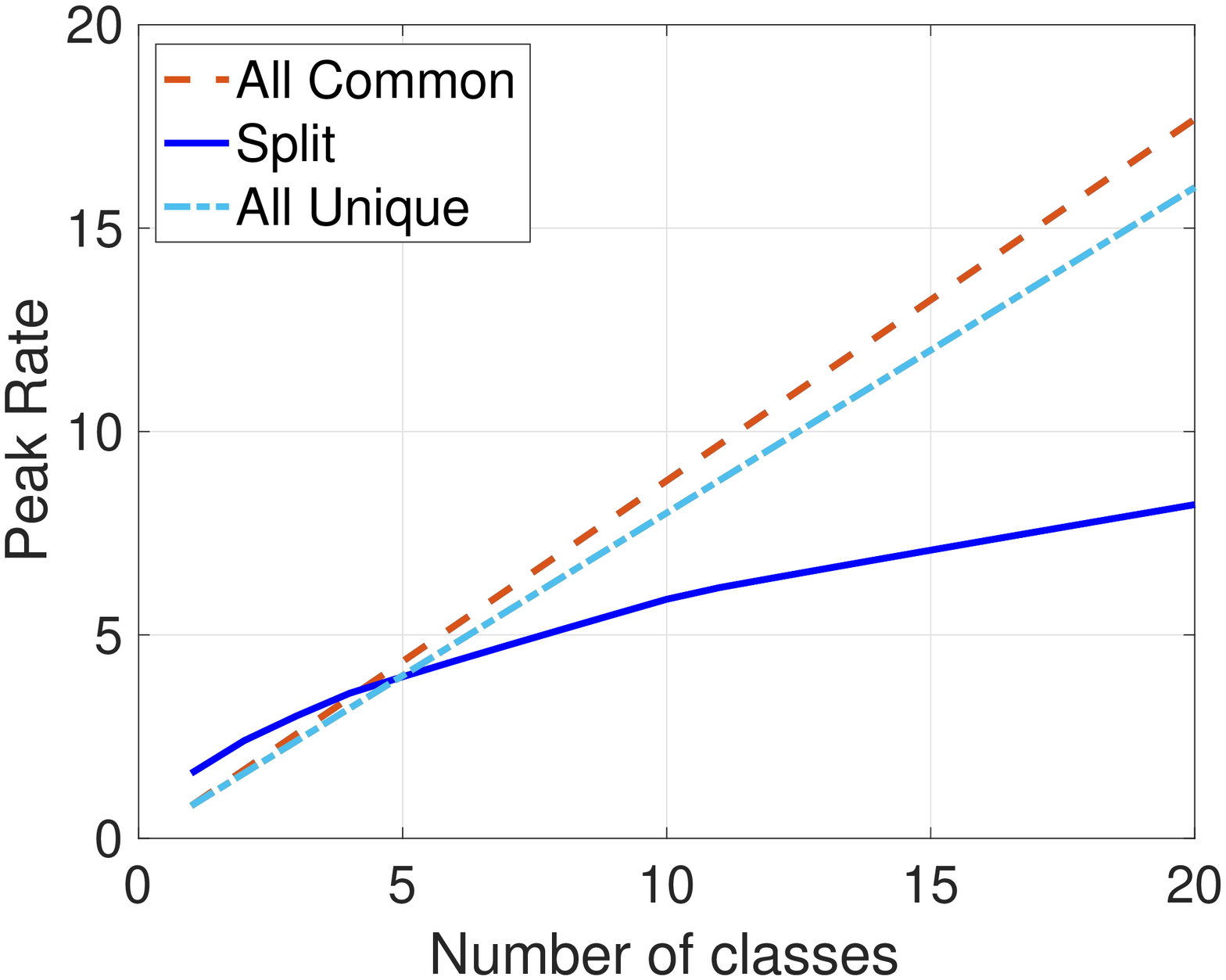}}
\centerline{\includegraphics[width=.4\textwidth]{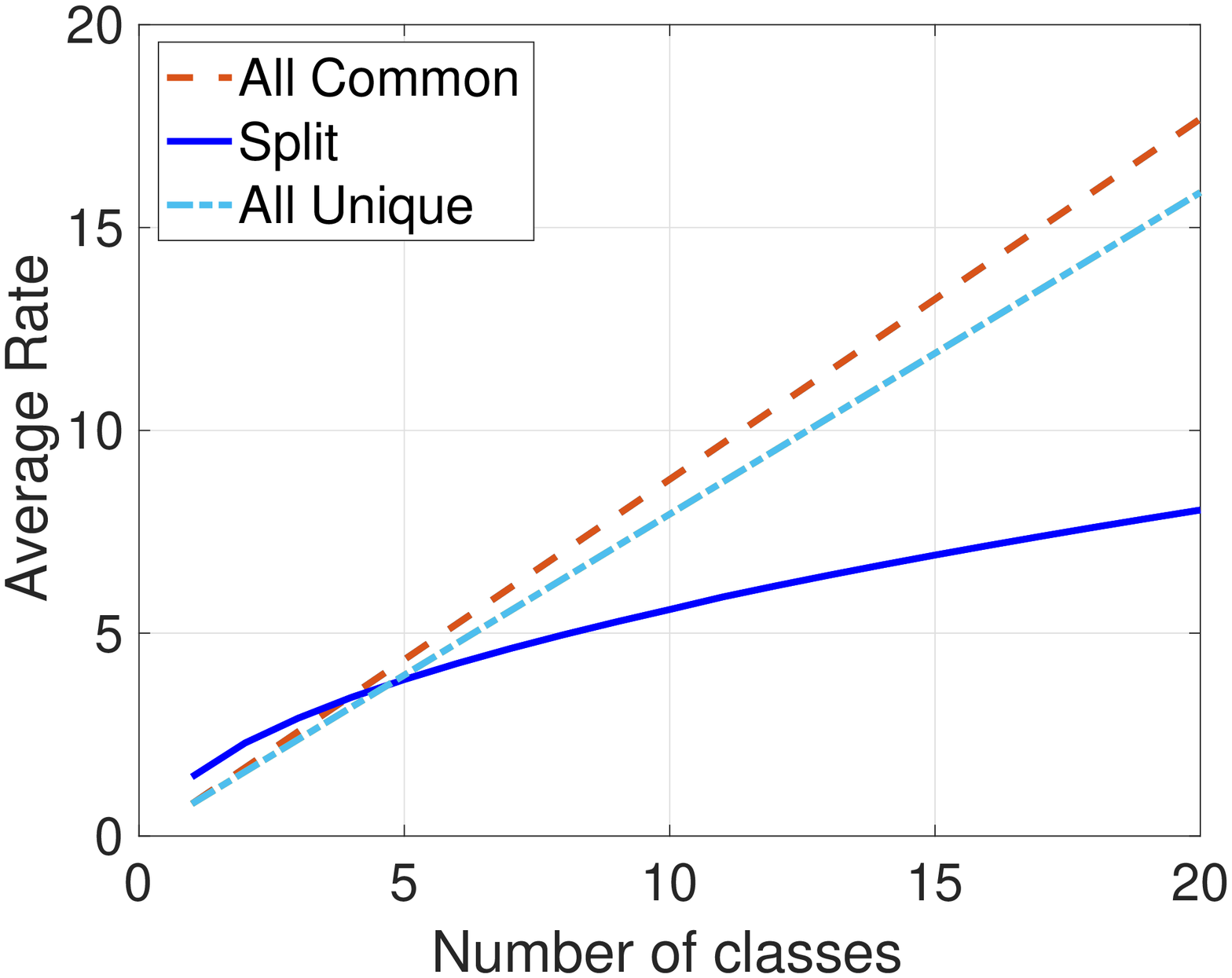}}
\end{center}
\caption{Top: Peak rate vs number of classes for $N_c=256$, $N_u=256$, $M=256$ and 8 users per class. \\
Bottom:  Average rate vs number of classes for $N_c=256$, $N_u=256$, $M=256$ and 8 users per class.}
\label{fig:avg_vs_class}
\end{figure}

\section{Conclusion}
\label{s-conclusion}
This paper proposes three coded caching schemes with uncoded pre-fetching which are suitable for a system where end users are categorized into classes according to their demand distributions. It is assumed that the files are either common, which means that they can be requested by any user in any class, or unique, meaning that only users in a specific class are likely to request them. The first scheme treats all files as if they were common, the second one decouples the delivery of common and unique files, and the third treats all files as if they were unique.

The peak and uniform-average rates of the three schemes are derived and compared with each other, showing that there exist conditions under which each scheme outperforms the other two. Specifically, Scheme~1 provides the lowest peak rate when the caches are small and Scheme~2 when the caches are large. Scheme~3 is best for intermediate cache sizes when the number of classes is small and the number of users is large. Their peak rates are also compared with a cut-set lower bound on the achievable rate to obtain bounds on the gap to optimality for each scheme.

In future work, we plan to study the uniform-average rate of the three schemes in more detail, seeking lower bounds on the minimal rates achievable and attempting to characterize the capacity region of coded caching with heterogeneous user profiles.



\bibliography{reference}

\begin{thebibliography}{10}
\providecommand{\url}[1]{#1}
\csname url@samestyle\endcsname
\providecommand{\newblock}{\relax}
\providecommand{\bibinfo}[2]{#2}
\providecommand{\BIBentrySTDinterwordspacing}{\spaceskip=0pt\relax}
\providecommand{\BIBentryALTinterwordstretchfactor}{4}
\providecommand{\BIBentryALTinterwordspacing}{\spaceskip=\fontdimen2\font plus
\BIBentryALTinterwordstretchfactor\fontdimen3\font minus
  \fontdimen4\font\relax}
\providecommand{\BIBforeignlanguage}[2]{{%
\expandafter\ifx\csname l@#1\endcsname\relax
\typeout{** WARNING: IEEEtran.bst: No hyphenation pattern has been}%
\typeout{** loaded for the language `#1'. Using the pattern for}%
\typeout{** the default language instead.}%
\else
\language=\csname l@#1\endcsname
\fi
#2}}
\providecommand{\BIBdecl}{\relax}
\BIBdecl

\bibitem{yi2015survey}
S.~Yi, C.~Li, and Q.~Li, ``A survey of fog computing: concepts, applications
  and issues,'' in \emph{Proceedings of the 2015 workshop on mobile big data},
  2015, pp. 37--42.

\bibitem{maddah2014fundamental}
M.~A. Maddah-Ali and U.~Niesen, ``Fundamental limits of caching,'' \emph{IEEE
  Transactions on Information Theory}, vol.~60, no.~5, pp. 2856--2867, 2014.

\bibitem{niesen2017coded}
U.~Niesen and M.~A. Maddah-Ali, ``Coded caching with nonuniform demands,''
  \emph{IEEE Transactions on Information Theory}, vol.~63, no.~2, pp.
  1146--1158, 2017.

\bibitem{pedarsani2016online}
R.~Pedarsani, M.~A. Maddah-Ali, and U.~Niesen, ``Online coded caching,''
  \emph{IEEE/ACM Transactions on Networking (TON)}, vol.~24, no.~2, pp.
  836--845, 2016.

\bibitem{chan2019coded}
C.-H. Chang and C.-C. Wang, ``Coded caching with full heterogeneity: Exact
  capacity of the two-user/two-file case,'' in \emph{2019 IEEE International
  Symposium on Information Theory (ISIT)}.\hskip 1em plus 0.5em minus
  0.4em\relax IEEE, 2019, pp. 6--10.

\bibitem{ibrahim2019coded}
A.~M. Ibrahim, A.~A. Zewail, and A.~Yener, ``Coded caching for heterogeneous
  systems: An optimization perspective,'' \emph{IEEE Transactions on
  Communications}, vol.~67, no.~8, pp. 5321--5335, 2019.

\bibitem{bayat2020cache}
M.~Bayat, K.~Wan, M.~Ji, and G.~Caire, ``Cache-aided modulation for
  heterogeneous coded caching over a gaussian broadcast channel,'' \emph{arXiv
  preprint arXiv:2001.05784}, 2020.

\bibitem{daniel2019optimization}
A.~M. Daniel and W.~Yu, ``Optimization of heterogeneous coded caching,''
  \emph{IEEE Transactions on Information Theory}, vol.~66, no.~3, pp.
  1893--1919, 2019.

\bibitem{8091300}
J.~{Zhang}, X.~{Lin}, and X.~{Wang}, ``Coded caching under arbitrary popularity
  distributions,'' \emph{IEEE Transactions on Information Theory}, vol.~64,
  no.~1, pp. 349--366, 2018.

\bibitem{golrezaei2013femtocaching}
N.~Golrezaei, A.~F. Molisch, A.~G. Dimakis, and G.~Caire, ``Femtocaching and
  device-to-device collaboration: A new architecture for wireless video
  distribution,'' \emph{IEEE Communications Magazine}, vol.~51, no.~4, pp.
  142--149, 2013.

\bibitem{karamchandani2016hierarchical}
N.~Karamchandani, U.~Niesen, M.~A. Maddah-Ali, and S.~N. Diggavi,
  ``Hierarchical coded caching,'' \emph{IEEE Transactions on Information
  Theory}, vol.~62, no.~6, pp. 3212--3229, 2016.

\bibitem{zhang2015coded}
J.~Zhang, X.~Lin, C.-C. Wang, and X.~Wang, ``Coded caching for files with
  distinct file sizes,'' in \emph{2015 IEEE International Symposium on
  Information Theory (ISIT)}.\hskip 1em plus 0.5em minus 0.4em\relax IEEE,
  2015, pp. 1686--1690.

\bibitem{hannak2013measuring}
A.~Hannak, P.~Sapiezynski, A.~Molavi~Kakhki, B.~Krishnamurthy, D.~Lazer,
  A.~Mislove, and C.~Wilson, ``Measuring personalization of web search,'' in
  \emph{Proceedings of the 22nd international conference on World Wide
  Web}.\hskip 1em plus 0.5em minus 0.4em\relax ACM, 2013, pp. 527--538.

\bibitem{mccallum1999machine}
A.~McCallum, K.~Nigam, J.~Rennie, and K.~Seymore, ``A machine learning approach
  to building domain-specific search engines,'' in \emph{IJCAI}, vol.~99.\hskip
  1em plus 0.5em minus 0.4em\relax Citeseer, 1999, pp. 662--667.

\bibitem{agichtein2006learning}
E.~Agichtein, E.~Brill, S.~Dumais, and R.~Ragno, ``Learning user interaction
  models for predicting web search result preferences,'' in \emph{Proceedings
  of the 29th annual international ACM SIGIR conference on Research and
  development in information retrieval}.\hskip 1em plus 0.5em minus 0.4em\relax
  ACM, 2006, pp. 3--10.

\bibitem{lu2019effective}
Y.~Lu, W.~Chen, and H.~V. Poor, ``On the effective throughput of coded caching:
  A game theoretic perspective,'' \emph{arXiv preprint arXiv:1911.12981}, 2019.

\bibitem{he2020coded}
J.~He, C.~Li, and L.~Song, ``Coded caching with heterogeneous user groups,'' in
  \emph{2020 IEEE Wireless Communications and Networking Conference Workshops
  (WCNCW)}.\hskip 1em plus 0.5em minus 0.4em\relax IEEE, 2020, pp. 1--6.

\bibitem{tegin2020coded}
B.~Tegin and T.~M. Duman, ``Coded caching with user grouping over wireless
  channels,'' \emph{IEEE Wireless Communications Letters}, 2020.

\bibitem{wang2019coded}
S.~Wang and B.~Peleato, ``Coded caching with heterogeneous user profiles,''
  \emph{IEEE Internat. Symp. on Information Theory (ISIT)}, 2019.

\bibitem{luo2019coded}
T.~Luo, V.~Aggarwal, and B.~Peleato, ``Coded caching with distributed
  storage,'' \emph{IEEE Transactions on Information Theory}, vol.~65, no.~12,
  pp. 7742--7755, 2019.

\bibitem{yu2018exact}
Q.~Yu, M.~A. Maddah-Ali, and A.~S. Avestimehr, ``The exact rate-memory tradeoff
  for caching with uncoded prefetching,'' \emph{IEEE Transactions on
  Information Theory}, vol.~64, no.~2, pp. 1281--1296, 2018.

\bibitem{chang2020twousers}
C.-H. Chang, C.-C. Wang, and B.~Peleato, ``On coded caching for two users with
  overlapping demand sets,'' \emph{IEEE International Conf. on Comm}, June
  2020.

\bibitem{zhang202averagerate}
C.~Zhang and B.~Peleato, ``Average rate for coded caching with heterogeneous
  user profiles,'' \emph{IEEE International Conf. on Comm}, June 2020.

\bibitem{yu2018characterizing}
Q.~Yu, M.~A. Maddah-Ali, and A.~S. Avestimehr, ``Characterizing the rate-memory
  tradeoff in cache networks within a factor of 2,'' \emph{IEEE Transactions on
  Information Theory}, vol.~65, no.~1, pp. 647--663, 2018.

\bibitem{luo2018transfer}
T.~Luo and B.~Peleato, ``The transfer load-i/o trade-off for coded caching,''
  \emph{IEEE Communications Letters}, vol.~22, no.~8, pp. 1524--1527, 2018.

\bibitem{tm1991elements}
C.~TM and T.~JA, ``Elements of information theory,'' \emph{Wiley Series in
  Telecommunications}, 1991.

\bibitem{chang2019coded}
C.-H. Chang and C.-C. Wang, ``Coded caching with heterogeneous file demand sets
  - the insufficiency of selfish coded caching,'' \emph{IEEE Internat. Symp. on
  Information Theory (ISIT)}, 2019.

\bibitem{jin2017structural}
S.~Jin, Y.~Cui, H.~Liu, and G.~Caire, ``Structural properties of uncoded
  placement optimization for coded delivery,'' \emph{arXiv preprint
  arXiv:1707.07146}, 2017.

\end{thebibliography}

\appendix

\subsection{Proof of Theorem~\ref{pr:schemeII_peak}}
\label{SecondthAppendix}

\begin{proof}

First observe that, according to the continuous relaxation of binomial coefficients in Eq.~(\ref{eq:fact_gamma_relaxation}),
\begin{equation}
    \frac{\partial}{\partial b}\binom{a}{b}=\binom{a}{b}\cdot[\psi(a-b+1)-\psi(b+1)],
\end{equation}
where $\psi(x)$ denotes the digamma function. This formula will significantly simplify the calculations throughout the proof.

The partial derivative of $R_\mathrm{peak}^{(2)}$ with respect to $x$ is,
\begin{equation}\label{eq:partial_schemeII_peak_general}
    \frac{\partial R_\mathrm{peak}^{(2)}}{\partial x}  = \frac{\partial R_c(x,\alpha)}{\partial x}+ \frac{\partial (GR_u(x,\alpha))}{\partial x},
\end{equation}
where
\begin{IEEEeqnarray}{rCl}
\frac{\partial R_c(x,\alpha)}{\partial x}&=  &\frac{1}{\binom{K}{t_c}}\frac{KM}{N_c} \Bigg\{ \Big[  \psi(t_c+1)- \psi(t_c+2) \Big]\cdot\left[ \binom{K}{t_c+1}-\binom{G\alpha }{t_c+1}\right]\nonumber  \\  &&+ \binom{K}{t_c+1}\cdot \Big[ \psi( K-t_c ) - \psi( K-t_c+1 ) \Big] + \binom{G\alpha }{t_c+1} \cdot\>\Big[ \psi( K-t_c+1 )- \psi( G\alpha -t_c) \Big] \Bigg\}\nonumber
\end{IEEEeqnarray}
and
\begin{IEEEeqnarray}{rCl}
\frac{\partial (GR_u(x,\alpha))}{\partial x} & = & \frac{1}{\binom{\frac{K}{G}}{t_u}}\frac{KM}{N_u} \Bigg\{ \Big[ \psi(t_u+2)\nonumber- \psi(t_u+1) \Big]  \cdot\>\left[ \binom{\frac{K}{G}}{t_u+1} -\binom{\frac{K}{G}-\alpha }{t_u+1}\right]\\ \nonumber &&+\binom{\frac{K}{G}}{t_u+1} \cdot\>\Big[ \psi\Big( \frac{K}{G}-t_u+1 \Big) - \psi\Big(\frac{K}{G}-t_u \Big) \Big] +\> \binom{\frac{K}{G}-\alpha}{t_u+1}\Big[\psi\Big( \frac{K}{G}-t_u- \alpha \Big)-\> \psi\Big(\frac{K}{G}-t_u+1 \Big) \Big]\Bigg\}.\nonumber
\end{IEEEeqnarray}

The digamma function has an interesting relationship with harmonic numbers. Specifically, $\psi(z+1)-\psi(z)=\frac{1}{z}$ for all positive $z$. As a consequence, 
it can be derived that
\begin{IEEEeqnarray}{rCl}
\frac{\partial R_c(x,\alpha)}{\partial x}&\geq Y_1\\ 
\frac{\partial (GR_u(x,\alpha))}{\partial x}&\leq Y_2, 
\end{IEEEeqnarray}
where
\begin{equation}
    Y_1 = - \frac{KM}{N_c}\frac{\binom{K}{t_c+1}}{\binom{K}{t_c}}\Big( \frac{1}{t_c+1} + \frac{1}{K-t_c} \Big),
\end{equation}
\begin{equation}
Y_2 = \frac{KM}{N_u}\frac{\binom{\frac{K}{G}}{t_u+1}}{\binom{\frac{K}{G}}{t_u}}\Big( \frac{1}{t_u+1} + \frac{1}{
\frac{K}{G}-t_u} \Big).
\end{equation}

It can be seen that $Y_1<0<Y_2$, $\forall{x}\in[0,1]$. $|Y_1|$ is minimal when $x=1$, $|Y_2|$ reaches maximum when $x=1$. Hence, we merely need to prove that $|Y_1|\geq|Y_2|$ when $x=1$, so that $\frac{\partial R_\mathrm{avg}^{(2)}}{\partial x}\leq 0,\forall x\in[0,1]$. The inequality can be written as
\begin{equation}
    M \leq\frac{N_c}{K} \left[ \sqrt{\frac{N_u(K+1)}{N_c(\frac{K}{G}+1)}}-1 \right].
\end{equation}
\end{proof}

\subsection{Proof of Theorem~\ref{pr:propI}}

\begin{proof}
When $M\leq\frac{1}{K}\min(N_c,GN_u)$, we are assured that both $t_c$ and $t_u$ in Eqs.~(\ref{eq:Rc})~and~(\ref{eq:Ru}) will be smaller than 1. As a consequence, the uniform-average rate must be adjusted according to Eq.~(\ref{eq:rate_t<1}) and $R_{\mathrm{avg}}^{(2)}$ becomes a linear function of $x$ (Eq.~(\ref{eq:rate_t<1}) is a linear function of $p$ and $p$ is a linear function of $x$). Since it is only defined over $0\leq x\leq 1$, $R_{\mathrm{avg}}^{(2)}$ must be minimized by either $x=0$ or $x=1$, depending on the sign of its partial derivative respect to $x$.

The law of large numbers tells us that when the number of files $K$ is large, the number of distinct files requested will be very close to its expected value for almost every demand vector $\vec{d}$. If we approximate $N_c(\vec{d})$ and $N_u(\vec{d})$ with their expected values, we have
\begin{align}
    \frac{\partial R_\mathrm{avg}^{(2)}}{\partial x}  \simeq& \frac{\partial R_c}{\partial t_c}\frac{\partial t_c}{\partial x} + G \frac{\partial R_u}{\partial t_u}\frac{\partial t_u}{\partial x}\\
    \simeq& \frac{KM}{N_c}\left\{ -\mathbb{E}[N_c(\vec{d})]+\frac{\binom{K}{2}-\binom{K-\mathbb{E}[N_c(\vec{d})]}{2}}{K}\right\} +\frac{KM}{N_u} \left\{\mathbb{E}[N_u(\vec{d})]-\frac{\binom{\frac{K}{G}}{2}-\binom{\frac{K}{G}-\mathbb{E}[N_u(\vec{d})]}{2}}{\frac{K}{G}}\right\}.
    \label{eq:R_2_t<1}
\end{align}
After expanding Eq.~(\ref{eq:R_2_t<1}) and cancelling out terms we find that $\frac{\partial R_\mathrm{avg}^{(2)}}{\partial x}$ is negative (\ie, $x=1$ minimizes $R_\mathrm{avg}^{(2)}$) when
\begin{equation}
    \frac{N_u}{N_c}> G\cdot\frac{\mathbb{E}^2[N_u(\vec{d})] + \mathbb{E}[N_u(\vec{d})] }{ \mathbb{E}^2[N_c(\vec{d})]+\mathbb{E}[N_c(\vec{d})] },
\end{equation}
and positive otherwise (\ie, $x=0$ minimizes $R_\mathrm{avg}^{(2)}$).

The expected number of distinct common and unique files requested can be computed using the method of indicators:
\begin{align}
\mathbb{E}[N_c(\vec{d})] &= N_c\left[ 1- \left( \frac{N_c+N_u-1}{N_c+N_u} \right)^K \right],\\
\mathbb{E}[N_{u}(\vec{d})] &= N_u\left[ 1- \left( \frac{N_c+N_u-1}{N_c+N_u} \right) ^{K/G} \right].
\end{align}
\end{proof}

\end{document}